\documentclass{bull-l}

\usepackage{amssymb}

\usepackage[cmtip,all]{xy}

\usepackage{graphicx,amscd}

\newtheorem{theorem}{Theorem}[section]

\newtheorem{proposition}[theorem]{Proposition}

\theoremstyle{definition}

\theoremstyle{remark}

\numberwithin{equation}{section}

\usepackage[cmtip,all]{xy}
\usepackage{color}

\newcommand{\D}{{\mathrm{d}}}

\newcommand{\aaa}{\mbox{\boldmath$a$}}

\newcommand{\kk}{\mbox{\boldmath$k$}}

\newcommand{\cc}{\mbox{\boldmath$c$}}

\newcommand{\e}{\mbox{\boldmath$e$}}

\newcommand{\ff}{\mbox{\boldmath$f$}}

\newcommand{\uu}{\mbox{\boldmath$u$}}

\newcommand{\qq}{\mbox{\boldmath$q$}}

\newcommand{\rr}{\mbox{\boldmath$r$}}

\newcommand{\vv}{\mbox{\boldmath$v$}}

\newcommand{\xx}{\mbox{\boldmath$x$}}

\newcommand{\II}{\mbox{\bf  1}}
\newcommand{\VV}{\mbox{\boldmath$V$}}

\newcommand{\tr}{\mbox{tr}}

\newcommand{\s}{\mbox{\boldmath$\sigma$}}

\begin{document}

\title[Hilbert's 6th Problem: Exact Hydrodynamics]{Hilbert's 6th Problem: Exact and Approximate Hydrodynamic Manifolds for Kinetic Equations}


\author{Alexander N. Gorban}
\address{Department of Mathematics, University of Leicester, LE1 7RH, Leicester, UK}
\curraddr{}
\email{ag153@le.ac.uk}
\thanks{}

\author{Ilya Karlin}
\address{Department of Mechanical and Process Engineering, ETH Z\"urich
                 8092 Z\"urich, Switzerland}
\curraddr{}
\email{karlin@lav.mavt.ethz.ch}
\thanks{}


\date{}

\dedicatory{}

\begin{abstract}
The problem of the derivation of hydrodynamics from the Boltzmann equation and related
dissipative systems is formulated as the problem of slow invariant manifold in the space
of distributions. We review a few instances where such hydrodynamic manifolds were found
analytically both as the result of summation of the Chapman--Enskog asymptotic expansion
and by the direct solution of the invariance equation. These model cases, comprising
Grad's moment systems, both linear and nonlinear, are studied in depth in order to gain
understanding of what can be expected for the Boltzmann equation. Particularly, the
dispersive dominance and saturation of dissipation rate of the exact hydrodynamics in the
short-wave limit and the viscosity modification at high divergence of the flow velocity
are indicated as severe obstacles to the resolution of Hilbert's 6th Problem.
Furthermore, we review the derivation of the approximate hydrodynamic manifold for the
Boltzmann equation using Newton's iteration and avoiding smallness parameters, and
compare this to the exact solutions. Additionally, we discuss the problem of projection
of the Boltzmann equation onto the approximate hydrodynamic invariant manifold using
entropy concepts. Finally, a set of hypotheses is put forward where we describe open
questions and set a horizon for what can be derived exactly or proven about the
hydrodynamic manifolds for the Boltzmann equation in the future.
\end{abstract}

\maketitle

\tableofcontents

\section{Introduction}
\subsection{Hilbert's 6th Problem}
The 6th Problem differs significantly from the other 22 Hilbert's problems
\cite{HilbertProblems}.  The title of the problem  itself is mysterious: ``Mathematical
treatment of the axioms of physics''. Physics, in its essence, is a special activity for
the creation, validation and destruction of theories for real-world phenomena, where ``We
are trying to prove ourselves wrong as quickly as possible, because only in that way can
we find progress'' \cite{FeynmanCHaracter}. There exist no mathematical tools to
formalize relations between Theory and Reality in live Physics. Therefore the 6th Problem
may be viewed as a tremendous challenge in deep study of ideas of physical reality in
order to replace vague philosophy by a new logical and mathematical discipline. Some
research  in quantum observation theory and related topics can be viewed as steps in that
direction, but it seems that at present we are far from an understanding of the most
logical and mathematical problems here.

The first explanation of the 6th Problem given by Hilbert reduced the level of challenge
and made the problem more tractable:  ``The investigations on the foundations of geometry
suggest the problem: To treat in the same manner, by means of axioms, those physical
sciences in which mathematics plays an important part; in the first rank are the theory
of probabilities and mechanics''. This is definitely ``a programmatic call''
\cite{Corry1997} for the axiomatization of the formal parts of {\em existent} physical
theories and no new universal logical framework for the representation of reality is
necessary. In this context, the axiomatic approach is a tool for the  retrospective
analysis of well-established and elaborated physical theories \cite{Corry1997} and not
for live physics.

For the general statements of the 6th Problem it seems unclear now how to formulate
criteria of solutions. In a further explanation Hilbert proposed two specific problems:
(i) axiomatic treatment of probability with limit theorems for foundation of statistical
physics and (ii) the rigorous theory of limiting processes ``which lead from the
atomistic view to the laws of motion of continua''. For complete resolution of these
problems Hilbert has set no criteria either but some important parts of them have been
already claimed as solved. Several axiomatic approaches to probability have been
developed and the equivalence of some of them has been proven \cite{Gnedenko1952}.
Kolmogorov's axiomatics (1933) \cite{Kolmogorov1933} is now accepted as standard. Thirty
years later, the complexity approach to randomness was invented by Solomonoff and
Kolmogorov (see the review \cite{ZvonkinLevin1970} and the textbook
\cite{LiVitannyi1997}). The rigorous foundation of equilibrium statistical physics of
many particles based on the central limit theorems was proposed
\cite{Khinchin1949,Dobrushin1977}. The modern development of the limit theorems in high
dimensions is based on the geometrical ideas of the measure concentration effects
\cite{Gromov2000,Talagrand1996} and gives new insights into the foundation of statistical
physics (see, for example, \cite{GorbanOrder2007,Talagrand2000}). Despite many open
questions, this part of the Hilbert programme is essentially fulfilled -- the probability
theory and the foundations of equilibrium statistical physics are now well-established
chapters of mathematics.

The way from the ``atomistic view to the laws of motion of continua'' is not so well
formalized. It includes at least two steps: (i) from mechanics to kinetics (from Newton
to Boltzmann) and (ii) from kinetics to mechanics and non-equilibrium thermodynamics of
continua (from Boltzmann to Euler and Navier--Stokes--Fourier).

The first part of the problem, the transition from the reversible--in--time equations of mechanics to
irreversible kinetic equations, is still too far from a complete rigorous theory. The highest achievement
here is the proof that rarefied gas of hard spheres will follow the Boltzmann equation
during a fraction of the collision time, starting from a non-correlated initial
state \cite{Lanford1975,GallagherS-R}. The BBGKY hierarchy \cite{Bogoliubov1946} provides the general
framework for this problem. For the systems close to global thermodynamic equilibrium
the global in time estimates are available and the validity of the linearized
Boltzmann equation is proven recently in this limit for rarefied gas of hard spheres \cite{BodineauS-R2013}.

The second part, model reduction in dissipative systems, from kinetics to macroscopic dynamics, is ready
for a mathematical treatment. Some limit theorems about this model reduction are already
proven (see the review book \cite{Saint-RaymondBook} and the companion paper by L.
Saint-Raymond \cite{Saint-RaymondCompanion} in this volume), and open questions can be
presented in a rigorous mathematical form. Our review is focused on this model reduction
problem, which is important in many areas of kinetics, from the Boltzmann equation to
chemical kinetics. There exist many similar heuristic approaches for different
applications \cite{GorKarLNP2005,Maas,RouFra,SeSlem}.

It seems that Hilbert presumed the kinetic level of description (the ``Boltzmann level'')
as an intermediate step between the microscopic mechanical description and the continuum mechanics. Nevertheless, this
intermediate description may be omitted. The transition from the microscopic to the macroscopic
description without an intermediate kinetic equation is used in many physical theories like
the Green--Kubo formalism \cite{Kubo1957}, the Zubarev method of a nonequilibrium
statistical operator \cite{Zubarev1974}, and the projection operator techniques \cite{Grabert1982}.
This possibility is demonstrated rigorously for a rarefied gas near global equilibrium \cite{BodineauS-R2013}.

The reduction from the Boltzmann kinetics to hydrodynamics may be split into three
problems: existence of hydrodynamics, the form of the hydrodynamic equations and the
relaxation of the Boltzmann kinetics to hydrodynamics. Formalization of these problems is
a crucial step in the analysis.

Three questions arise:
\begin{enumerate}
\item{Is there hydrodynamics in the kinetic equation, i.e., is it possible to lift
    the hydrodynamic fields to the relevant one-particle distribution functions in
    such a way that the projection of the kinetics of the relevant distributions
    satisfies some hydrodynamic equations?}
\item{Do these hydrodynamics have the conventional Euler and Navier--Stokes--Fourier
    form?}
\item{Do the solutions of the kinetic equation degenerate to the hydrodynamic regime
    (after some transient period)?}
\end{enumerate}

The first question is the problem of existence of a hydrodynamic invariant manifold for
kinetics (this manifold should be parameterized by the hydrodynamic fields). The second
one is about the form of the hydrodynamic equations obtained by the natural projection of
kinetic equations from the invariant manifold. The third question is about the
intermediate asymptotics of the relaxation of kinetics to equilibrium: do the solutions
go fast to the hydrodynamic invariant manifold and then follow this manifold on the path
to equilibrium?

The answer to all three questions is essentially positive in the asymptotic regime when
the Mach number $M\!a$ and the Knudsen number $K\!n$ tend to zero
\cite{BardosGolLever1991,GolseS-R2004} (see
\cite{Saint-RaymondBook,Saint-RaymondCompanion}). This is a limit of very slow flows with
very small gradients of all fields, i.e. almost no flow at all. Such a flow changes in
time very slowly and  a rescaling of time $t_{\rm old} = t_{\rm new}/\varepsilon$ is
needed to return it to non-trivial dynamics (the so-called diffusive rescaling). After
the rescaling, we approach in this limit the Euler and Navier--Stokes--Fourier
hydrodynamics of {\em incompressible} liquids.

Thus in the limit $M\!a,K\!n \to 0$ and after rescaling the 6th Hilbert Problem is
essentially resolved and the result meets Hilbert's expectations: the continuum equations
are rigorously derived from the Boltzmann equation. Besides the limit the answers are
known partially. To the best of our knowledge, now the answers to these three questions
are: (1) sometimes; (2) not always; (3) possibly.

Some hints about the problems with hydrodynamic asymptotics can be found in the series of
works about the small dispersion limit of the Korteweg--de Vries equation
\cite{LaxLever1983}. Recently, analysis of the exact solution of the model reduction
problem for a simple kinetic model \cite{GKPRL96,SlemCAMWA2013} has demonstrated that a
hydrodynamic invariant manifold may exist and produce non-local hydrodynamics. Analysis
of more complicated kinetics
\cite{KGAnPh2002,JPA00,KarColKroe2008PRL,ColKarKroe2007_1,ColKarKroe2007_2} supports and
extends these observations: the hydrodynamic invariant manifold may exist but sometimes
does not exist, and the hydrodynamic equations when $M\!a \nrightarrow 0$ may differ
essentially from the Euler and Navier--Stokes--Fourier equations.

At least two effects prevent us from  giving positive answers to the first two  questions
outside of the limit $M\!a, K\!n \to 0$:
\begin{itemize}
\item{Entanglement between the hydrodynamic and non-hydrodynamic modes may destroy
    the hydrodynamic invariant manifold.}
\item{Saturation of dissipation at high frequencies is a universal effect that is
    impossible in the classical hydrodynamic equations.}
\end{itemize}
These effects appear already in simple linear kinetic models and are studied in detail
for the exactly solvable reduction problems. The {\em entanglement between the
hydrodynamic and non-hydrodynamic modes} manifests itself in many popular moment
approximations for the Boltzmann equation. In particular, it exists for the
three-dimensional 10-moment and 13-moment Grad systems
\cite{JPA00,KGAnPh2002,GorKarLNP2005,ColKarKroe2007_1,ColKarKroe2007_2} but the numerical
study of the hydrodynamic invariant manifolds for the BGK model equation
\cite{KarColKroe2008PRL} demonstrates the absence of such an entanglement. Therefore, our
conjecture is that for the Boltzmann equation the exact hydrodynamic modes are separated
from the non-hydrodynamic ones if the linearized collision operator has a spectral gap
between the five times degenerated zero and the rest of the spectrum.

The {\em saturation of dissipation} seems to be a universal phenomenon
\cite{RosenauSat1989,GorKarTTSPreg1993,GKTTSP94,KurganovRosenau1997,Slem2,KGAnPh2002,GorKarLNP2005}.
It appears in all exactly solved reduction problems for kinetic equations
\cite{KGAnPh2002} and in the Bhatnagar--Gross--Krook \cite{BGK}  (BGK) kinetics
\cite{KarColKroe2008PRL} and is also proven for various regularizations of the
Chapman--Enskog expansion \cite{RosenauSat1989,GorKarTTSPreg1993,Slem2,GorKarLNP2005}.

The answer to Hilbert's 6th  Problem concerning transition from the Boltzmann equation to
the classical equations of motion of compressible continua ($M\!a \nrightarrow 0$) may
turn out to be negative. Even if we can overcome the first difficulty, separate the
hydrodynamic modes from the non-hydrodynamic ones (as in the exact solution
\cite{GKPRL96} or for the BGK equation \cite{KarColKroe2008PRL}) and produce the
hydrodynamic equations from the Boltzmann equation, the result will be manifestly
different from the conventional equations of hydrodynamics.

\subsection{The main equations\label{Sec:MainEq}}

We discuss here two groups of examples. The first of them consists of kinetic equations
which describe the evolution of a one-particle gas distribution function $f(t, \xx; \vv)$
\begin{equation}\label{BOL0}
    \partial_t f+ \vv \cdot \nabla_x f= \frac{1}{\epsilon}Q(f),
\end{equation}
where $Q(f)$ is the collision operator. For the Boltzmann equation, $Q$ is a quadratic
operator and, therefore, the notation $Q(f,f)$ is often used.

The second group of examples are the systems of Grad moment equations
\cite{Grad,Bob,EIT,GorKarLNP2005}. The system of 13-moment Grad equations linearized near
equilibrium is
\begin{equation}
\begin{split}
\label{balanceequations}
\partial_t \rho &=-\nabla\cdot{\uu},\\
\partial_t {\uu} &=-\nabla \rho-\nabla T -\nabla\cdot
\s ,\\
\partial_t T &=-\frac{2}{3}(\nabla\cdot{\uu}+\nabla\cdot{\qq})
,
\end{split}
\end{equation}
\begin{equation}
\label{Grad133e}
\begin{split}
\partial_t \s &=-2\overline{\nabla {\uu}}-\frac{4}{5}
\overline{\nabla {\qq}}-\frac{1}
{\epsilon}\s,\\
\partial_t{\qq}&=-\frac{5}{2}\nabla T-\nabla\cdot\s -\frac{2}{3\epsilon}{\qq}.
\end{split}
\end{equation}

In these equations, $\s({\xx},t)$ is the dimensionless stress tensor, $\s = (\sigma_{ij})$,
and ${\qq({\xx},t)}$ is the dimensionless vector of heat flux, $\qq=(q_i)$. We use the
system of units in which Boltzmann's constant $k_{\rm B}$ and the particle mass $m$ are
equal to one, and the system of dimensionless variables:
\begin{equation}
\label{AnnPhysvar} {\uu}=\frac{\delta{\uu}}{\sqrt{T_0}},\ \rho=\frac{\delta\rho}{\rho_0},\
T=\frac{\delta T}{T_0}, \xx=\frac{\rho_0}{\eta(T_0)\sqrt{T_0}}\xx^{\prime}, \
t=\frac{\rho_0}{\eta(T_0 )}t^{\prime},
\end{equation}
where ${\xx}^{\prime}$ are spatial coordinates, and $t^{\prime}$ is time.

The dot denotes the standard scalar product, while the overline indicates the symmetric
traceless part of a tensor. For a tensor $\aaa=(a_{ij })$ this part is
$$\overline{\aaa}=\frac{1}{2}(\aaa + \aaa^T)- \frac{1}{3}{I} \tr(\aaa), $$ where ${I}$ is unit matrix.
In particular,
$$
\overline {{\nabla}{\uu}}=\frac{1}{2}({\nabla}{\uu}+({\nabla}{\uu})^T
-\frac{2}{3}{I}{\nabla}\cdot{\uu}),
$$
where $I=(\delta_{ij}$ is the identity matrix.

We also study a simple model of a coupling of the hydrodynamic variables, ${\uu}$ and $p$
($p(\xx,t)=\rho(\xx,t)+T(\xx,t)$), to the non-hydrodynamic variable ${\s}$, the 3D
linearized Grad equations for 10 moments $p$, ${\uu}$, and $\s $:
\begin{equation}
\label{Grad103}
\begin{split}
\partial_t p &=-\frac{5}{3}\nabla \cdot {\uu},\\
\partial_t {\uu} &=-\nabla p -\nabla\cdot\s ,
\\
\partial_t  \s &=-2\overline{\nabla{\uu}}
-\frac{1}{\epsilon}\s.
\end{split}
\end{equation}
Here, the coefficient $\frac{5}{3}$ is the adiabatic exponent of the 3D ideal gas.

The simplest model and the starting point in our analysis is the reduction of the system
(\ref{Grad103}) to the functions that depend on one space coordinate $x$ with the
velocity $\uu$ oriented along the $x$ axis:
\begin{equation}
\label{Grad101}
\begin{split}
\partial_t p &=-\frac{5}{3}\partial_x u,\\
\partial_t u &=-\partial_x p -\partial_x \sigma,\\
\partial_t \sigma &=-\frac{4}{3}\partial_x u
-\frac{1}{\epsilon}\sigma ,
\end{split}
\end{equation}
where $\sigma$ is the dimensionless $xx$-component of the stress tensor and the equation
describes the unidirectional solutions of the previous system (\ref{Grad103}).

These equations are elements of the staircase of simplifications, from the Boltzmann
equation to moment equations of various complexity, which was introduced by Grad
\cite{Grad} and elaborated further by many authors. In particular, Levermore proved hyperbolicity
of the properly constructed moment equations \cite{LeverMoments}.
This staircase forms the basis of the
Extended Irreversible Thermodynamics (EIT \cite{EIT}).

\subsection{Singular perturbation and separation of times in kinetics\label{Sec:TimeSep}}

The kinetic equations are singularly perturbed with a small parameter $\epsilon$ (the
``Knudsen number'') and we are interested in the asymptotic properties of solutions when
$\epsilon$ is small. The physical interpretation of the Knudsen number is the ratio of
the ``microscopic lengths'' (for example, the mean free path) to the ``macroscopic
scale'', where the solution changes significantly. Therefore, its definition depends on
the properties of solutions. If the space derivatives are uniformly bounded, then we can
study the asymptotic behavior $\epsilon \to 0$. But for some singular solutions this
problem statement may be senseless. The simple illustration of  rescaling with the
erasing of $\epsilon$ gives the set of travelling automodel  solutions for  (\ref{BOL0}).
If we look for them in a form $f=\varphi (\boldsymbol{\xi},\vv)$ where
$\boldsymbol{\xi}=(\xx-\cc t)/\epsilon$ then the equation for $\varphi
(\boldsymbol{\xi},\vv)$ does not depend on $\epsilon$:

$$    (\vv-\cc)\cdot \nabla_{\xi} \phi= Q(\phi).$$

In general, $\epsilon$ may be considered as a variable that is neither small nor large and the problem is to analyze the dependence of
solutions on $\epsilon$.

For the Boltzmann equation (\ref{BOL0}) the collision term $Q(f)$ does not enter directly
into the time derivatives of the hydrodynamic variables, $\rho=\int f \D \vv$, $\uu =
\int \vv f \D \vv$ and $T= \int (\vv-\uu)^2 f \D \vv$ due to the mass, momentum
and energy conservation laws
$$\int \{1; \vv; (\vv-\uu)^2 \} Q(f) \D \vv =0.$$
The following dynamical system point of view is valid  for smooth solutions in a bounded
region with no-flux and equilibrium boundary conditions, but it is used with some success
much more widely. The collision term is ``fast'' (includes the large parameter
$1/\varepsilon$) and does not affect the macroscopic hydrodynamic variables directly.
Therefore, the following qualitative picture is expected for the solutions: (i) the
collision term goes quickly almost to its equilibrium (the system almost approaches a
local equilibrium) and during this fast initial motion the changes of hydrodynamic
variables are small, (ii) after that the distribution function is defined with high
accuracy by the hydrodynamic variables (if they have bounded space derivatives). The
relaxation of the collision term almost to its equilibrium is supported by monotonic
entropy growth (Boltzmann's $H$-theorem). This qualitative picture is illustrated in
Fig.~\ref{fig1SLw}.

Such a ``nonrigorous picture of the Boltzmann dynamics'' \cite{DescvilVill2005} which
operates by the manifolds in the space of probability distributions is a seminal tool for
production of qualitative hypotheses. The points (`states') in Fig.~\ref{fig1SLw}.
correspond to the distributions $f(\xx,\vv)$, and the points in the projection correspond to the
hydrodynamic fields in space.

\begin{figure}
\begin{centering}
\boxed{\includegraphics[width=0.8 \textwidth]{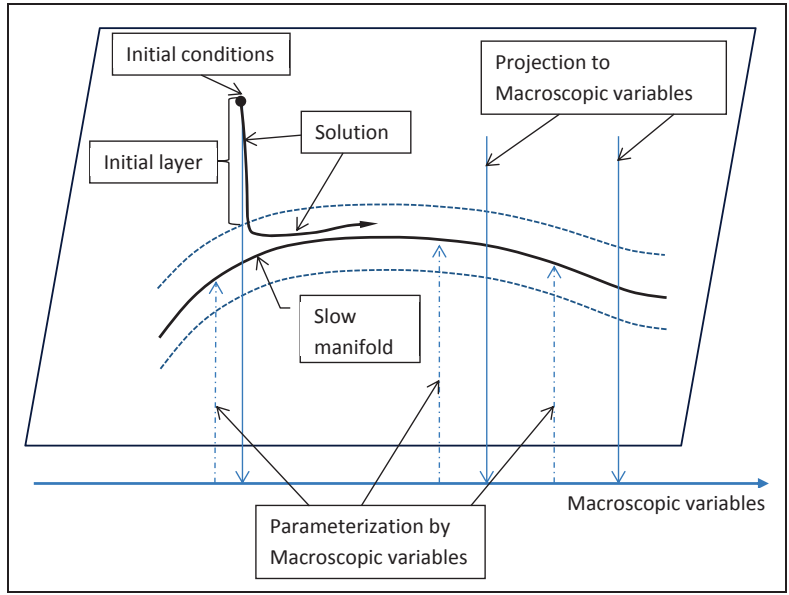}}
\caption {\label{fig1SLw} Fast--slow decomposition. Bold dashed lines outline the
vicinity of the slow manifold where the solutions stay after initial layer.
The projection of the distributions onto the hydrodynamic fields and the
parametrization of this manifold by the hydrodynamic fields are represented. }
\end{centering}
\end{figure}

For the Grad equations (\ref{balanceequations})-(\ref{Grad133e}), (\ref{Grad103}) and
(\ref{Grad101}) the hydrodynamic variables $\rho, \uu, T$ are explicitly separated from
the fluxes and the projection onto the hydrodynamic fields is just the selection of the
hydrodynamic part of the set of all fields. For example, for (\ref{Grad101}) this is just
the selection of  $p(\xx), \uu(\xx)$ from the  whole set of fields $p(\xx), \uu(\xx),
\boldsymbol{\sigma} (\xx)$. The expected qualitative picture for smooth solutions is the
same as in Fig.~\ref{fig1SLw}.

For finite-dimensional ODEs, Fig.~\ref{fig1SLw} represents the systems which satisfy the
Tikhonov singular perturbation theorem \cite{Tikhonov1952}. In some formal sense, this picture
for the Boltzmann equation is also rigorous when $\epsilon \to 0$
and is proven in \cite{BardosGolLever1991}. Assume that $f^{\epsilon}(t,\xx,\vv)$ is
a sequence of nonnegative solutions of the Boltzmann equation (\ref{BOL0}) when $\epsilon \to 0$ and there exists
a limit $f^{\epsilon}(t,\xx,\vv) \to f^0(t,\xx,\vv)$. Then (under some additional regularity conditions), this limit
$ f^0(t,\xx,\vv)$ is a local Maxwellian and the corresponding moments satisfy the compressible Euler equation. According to
\cite{Saint-RaymondBook}, this is ``the easiest of all hydrodynamic limits of the Boltzmann
equation at the formal level''.

The theory of singular
perturbations was developed starting from complex systems, from the Boltzmann equation (Hilbert
\cite{HilbertKinetishe}, Enskog \cite{Ens}, Chapman \cite{Chapman}, Grad
\cite{Grad,GradENC}) to ODEs. The recently developed geometric theory of singular
perturbation \cite{Fenichel0,Fenichel9,JonesGSP1995} can be considered as a formalization
of the Chapman--Enskog approach for the area where complete rigorous theory is
achievable.

A program  of the derivation of (weak) solutions of the Navier--Stokes equations from the (weak) solutions of the Boltzmann equation was formulated
in 1991 \cite{BardosGolLever1991} and finalized in 2004 \cite{GolseS-R2004} with the answer: the incompressible Navier--Stokes (Navier--Stokes--Fourier) equations appear in a limit of appropriately scaled solutions of the Boltzmann equation.

We use the geometry of time-separation (Fig.~\ref{fig1SLw}) as a guide for formal
constructions and present further development of this scheme using some ideas from
thermodynamics and dynamics.

\subsection{The structure of the paper\label{Sec:Structure}}

In  Sec.~\ref{Sec:InvEq} we introduce the invariance equation for invariant manifolds. It
has been studied by Lyapunov (Lyapunov's auxiliary theorem \cite{Lya1992},
Theorem~\ref{LyaAuxTHeo} below). We describe the structure of the invariance equations
for the Boltzmann and Grad equations and in Sec.~\ref{Sec:ChEnsk} construct the
Chapman--Enskog expansion for the solution of the invariance equation.

It may be worth stressing that the invariance equation is a {\em nonlinear equation} and
there is no known general method to solve them even for {\em linear} differential
equations. The main construction is illustrated on the simplest kinetic equation
(\ref{Grad101}): in Sec.~\ref{Sec:SimplestEuler...SuperBur} the Euler, Navier--Stokes,
Burnett, and super--Burnett terms are calculated for this equation and the ``ultraviolet
catastrophe'' of the Chapman--Enskog series is demonstrated (Fig.~\ref{Attenuation}).

The first example of the exact summation of the Chapman--Enskog series is presented in
detail for the simplest system (\ref{Grad101}) in  Sec.~\ref{Sec:Exact}. We analyze the
structure of the Chapman--Enskog series and find the pseudodifferential representation of
the stress tensor on the hydrodynamic invariant manifold. Using this representation, in
Sec.~\ref{Sec:caplillarity} we represent the energy balance equation in the
``capillarity--viscosity'' form proposed by Slemrod \cite{SlemCAMWA2013}. This form
explains the macroscopic sense of the dissipation saturation effect: the attenuation rate
does not depend on the wave vector $k$ for short waves (it tends to a constant value when
$k^2 \to \infty$). In the highly non-equilibrium gas the capillarity energy becomes
significant and it tends to infinity for high velocity gradients.

In the Fourier representation, the invariance equation for (\ref{Grad101}) is a system of
two coupled quadratic equations with linear in $k^2$ coefficients (Sec.~\ref{IEFourier}).
It can be solved in radicals and the corresponding hydrodynamics has the acoustic waves
decay with saturation (Sec.~\ref{stabSat}). The hydrodynamic invariant manifold for
(\ref{Grad101}) is analytic at the infinitely-distant point  $k^2=\infty$. Matching of
the first terms of the Taylor series in powers of $1/k^2$ with the first terms of the
Chapman--Enskog series gives simpler hydrodynamic equations with qualitatively the same
effects and even quantitatively the same saturation level of attenuation of acoustic
waves (Sec.~\ref{HighFreqAs}). We may guess that the matched asymptotics of this type
include all the essential information about hydrodynamics both at low and high
frequencies.

The construction of the invariance equations in the Fourier representation remains the
same for a general linear kinetic equation (Sec.~\ref{GenLinIE}). The exact hydrodynamics
on the invariant manifolds always inherits many important properties of the original
kinetics, such as dissipation and conservation laws. In particular,  if the original
kinetic system is hyperbolic then for bounded hydrodynamic invariant manifolds the
hydrodynamic equations are also hyperbolic (Sec.~\ref{Sec:Hyperbol}).

In Sec.~\ref{Sec:ExactLinKin}, we study the invariance equations for three systems: 1D
solutions of the 13 moment Grad system (Sec.~\ref{Sec:GradDestroy}), the full 3D 13
moment Grad system (Sec.~\ref{Sec:Grad3DDestroy}), and the linearized BGK kinetic
equation (Sec.~\ref{Sec:BGKexist}). The 13 moment Grad system demonstrates an important
effect: the invariance equation may lose the physically meaningful solution for short
waves. Therefore, existence of the exact hydrodynamic manifold is not compulsory for all
the usual kinetic equations. Nevertheless, for the BGK equation with the complete
advection operator $\vv \cdot \nabla$ the invariance equation exists for short waves too
(as is  demonstrated  numerically in \cite{KarColKroe2008PRL}).

For nonlinear kinetics, the exact solutions to the invariance equations are not known. In
Sec.~\ref{Sec:NlinIM} we demonstrate two approaches to approximate invariant manifolds.
First, for the nonlinear Grad equation we find  the leading terms of the Chapman--Enskog
series in the order of the Mach number and exactly sum  them. For this purpose, we
construct the approximate invariant manifold and find the solution for the nonlinear
viscosity in the form of an ODE (Sec.~\ref{Sec:NlinVis}). For the 1D solutions of the
Boltzmann equation we construct the invariance equation and demonstrate the result of the
first Newton--Kantorovich iteration for the solution of this equation
(Sec.~\ref{IMBoltzmann} and \cite{GKTTSP94,GorKarLNP2005}). Use of the {\em approximate}
invariant manifolds causes a problem of dissipativity preservation in the hydrodynamics
on these manifolds. There exists a unique modification of the projection operator that
guarantees the preservation of entropy production for hydrodynamics produced by
projection of kinetics onto an approximate invariant manifold even for rough
approximations \cite{UNIMOLD}. This construction is presented in
Sec.~\ref{Sec:Projector}. In Conclusion, we discuss solved and unsolved problems and
formulate several hypotheses.

\section{Invariance equation and Chapman--Enskog expansion\label{Sec:InvEq}}

\subsection{The idea of invariant manifold in kinetics\label{Sec:IdeaIMkinetics}}

Very often, the Chapman--Enskog expansion for the Boltzmann equation is introduced as an
asymptotic expansion in powers of $\epsilon$ of the solutions of equation
(\ref{BOL0}), which should depend on time only through time dependence of the macroscopic
hydrodynamic fields. Historically, the definition of the method is ``procedure
oriented'': an expansion is created step by step with the leading idea that solutions
should depend on time only through the macroscopic variables  and their derivatives. In
this approach what we are looking for often remains  hidden.

The result of the Chapman--Enskog method is not a solution  of the kinetic equation but
rather the proper parametrization of microscopic variables (distribution functions) by
the macroscopic (hydrodynamic) fields. It is a {\em lifting procedure}: we take the
hydrodynamic fields and find for them the corresponding distribution function. This
lifting should be consistent with the kinetics, i.e. the set of the corresponding
distributions (collected for all possible hydrodynamic fields) should be invariant with
respect to a shift in time. Therefore, the Chapman--Enskog procedure looks for an
invariant manifold for the kinetic equation which is close to the local equilibrium for a
small Knudsen number and smooth hydrodynamic fields with bounded derivatives. This is the
``object oriented'' description of the Chapman--Enskog procedure.

The puzzle in the statement of the problem of transition from kinetics to hydrodynamics
has been so deep that Uhlenbeck called it the ``Hilbert paradox'' \cite{Uhlenbeck1980}.
In the reduced hydrodynamic description, the state of a gas is completely determined if
one knows initially the space dependence of the five macroscopic variables $p$, $\uu$,
and $T$. Uhlenbeck has found this impossible: ``On the one hand it couldn't be true,
because the initial-value problem for the Boltzmann equation (which supposedly gives a
better description of the state of the gas) requires the knowledge of the initial value
of the distribution function $f(\rr,\vv,t)$ of which $p$, $\uu$, and $T$ are only the
first five moments in v. But on the other hand the hydrodynamical equations surely give a
causal description of the motion of a fluid. Otherwise how could fluid mechanics be
used?''

Perhaps, McKean gave the first clear explanation of the problem as a construction of a
`nice submanifold' where `the hydrodynamical equations define the same flow as the (more
complicated) Boltzmann equation does' \cite{McKean1965}. He presented the problem by a
partially commutative diagram and we use this idea in slightly revised form in
Fig.~\ref{McKeanDiag}.

\begin{figure}
\begin{centering}
\boxed{\includegraphics[width=0.8 \textwidth]{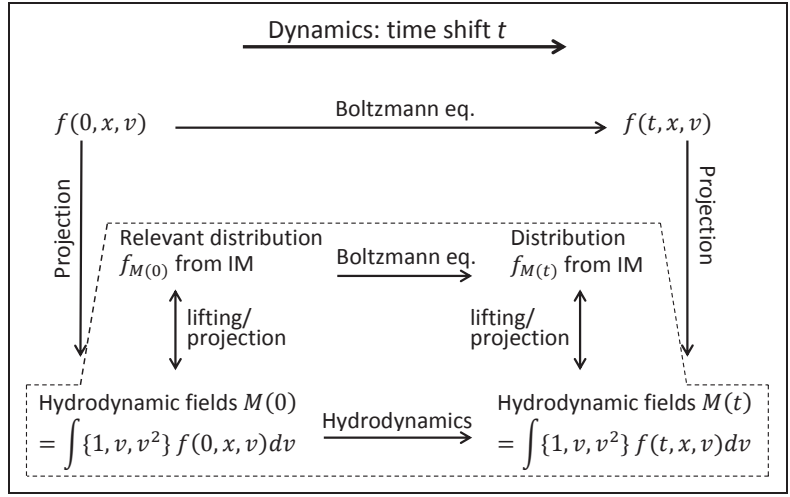}}
\caption {\label{McKeanDiag}{\em McKean diagram.} The Chapman--Enskog procedure aims to
create a lifting operation, from the hydrodynamic variables to
the corresponding distributions on the invariant manifold. IM stands for Invariant Manifold.
The part of the diagram in the dashed polygon is commutative.}
\end{centering}
\end{figure}

The invariance equation just expresses the fact that the vector field is tangent to the
manifold. The invariance equation has the simplest form for manifolds parameterized by
moments, i.e. by the values of the given linear functionals. Let us consider an equation
in a domain $U$ of a normed space $E$ with analytical (at least, Gateaux-analytical)
right hand sides
\begin{equation}\label{AbstractEq}
\partial_t f = J(f).
\end{equation}
A space of macroscopic variables (moment fields) is defined with a surjective linear map
to them $m: f \mapsto M$ ($M$ are macroscopic variables). Below when referring to a
manifold parameterized with the macroscopic fields $M$ we use the notation $\ff_M$. We
are looking for an invariant manifold  $\ff_{M}$ parameterized by the value of $M$, with
the self-consistency condition $m(\ff_{M})=M$.

The invariance equation is
\begin{equation}\label{IM}
\boxed
{J(\ff_M)=(D_M \ff_M) m(J(\ff_M)).}
\end{equation}
Here, the differential $D_M$ of $\ff_M$ is calculated at the point $M=m(\ff_M)$.

Equation (\ref{IM}) means that the time derivative of $\ff$ on the manifold $\ff_M$ can be calculated by a simple chain rule:
calculate the derivative of $M$ using the map $m$, $\dot{M}=m(J(\ff_M))$, and then write that the time dependence of $\ff$
can be expressed through the time dependence of $M$.
If we find the approximate solution to eq. (\ref{IM}) then the approximate reduced model (hydrodynamics) is
\begin{equation}\label{projectedMIM}
\partial_t M=m(J(\ff_M)) .
\end{equation}
The invariance equation can be represented in the form
$$\partial^{\rm micro}_t \ff_M = \partial^{\rm macro}_t \ff_M,$$
where the microscopic time derivative, $\partial^{\rm micro}_t \ff_M$ is just a value of
the vector field $J(\ff_M)$ and the macroscopic time derivative is calculated by the
chain rule, $$\partial^{\rm macro}_t \ff_M= (D_M \ff_M) \partial_t M $$ under the
assumption that dynamics of $M$ follows the projected  equation (\ref{projectedMIM}).

We use the natural (and naive) moment-based projection (\ref{projectedMIM}) till
Sec.~\ref{Sec:Projector} where we demonstrate that in many situations the modified
projectors are more suitable from thermodynamic point of view. In addition, the flexible
choice of projectors allows us to treat various nonlinear functionals (like scattering
rates) as macroscopic variables \cite{GorKarPRE1996,GKPhysA2006}.

If $\ff_M$ is a solution to the invariance equation (\ref{IM})  then the reduced model
(\ref{projectedMIM}) has two important properties:
\begin{itemize}
\item{{\bf Preservation of conservation laws.} If a differentiable functional $U(f)$
    is conserved due to the initial kinetic equation (\ref{AbstractEq}) then the
    functional $U_M=U(\ff_M)$ conserves due to reduced system  (\ref{projectedMIM}),
    i.e. it has zero time derivative due to this system.}
\item{{\bf Preservation of dissipation.} If the time derivative of a differentiable
    functional $H(f)$ is non-positive due to the initial kinetic equation, then the
    time derivative of the functional $H_M=H(\ff_M)$ is also
 non-positive due to reduced system.}
 \end{itemize}
These elementary properties are the obvious consequences of the invariance equation
(\ref{IM}) and the chain rule for  differentiation. Indeed, for every differentiable
functional $S(f)$ we introduce a functional $S_M=S(\ff_M)$. Then for the time derivative
of $S_M$ due to projected equation (\ref{projectedMIM}) coincides with the time
derivative of $S(f)$ at point $f=\ff_M$ due to (\ref{AbstractEq}).  (Preservation of time
derivatives.) Despite the very elementary character of these properties, they may be
extremely important in the construction of the energy and entropy formulas for the projected
equations (\ref{projectedMIM}) and in the proof of the $H$-theorem and hyperbolicity.

The difficulties with preservation of conservation laws and dissipation inequalities may
occur when one uses the approximate solutions of the invariance equation. For these
situations, two techniques are invented: modification of the projection operation (see
\cite{GK1,UNIMOLD} and Sec.~\ref{Sec:Projector} below) and modification of the entropy
functional \cite{GrmelaCAMWA2013,GrmelaAdv2010}. They allow to retain the dissipation
inequality for the approximate equations.

It is obvious that the invariance equation (\ref{IM}) for dynamical systems usually has
too many solutions, at least locally, in a vicinity of any non-singular point. For
example, every trajectory of (\ref{AbstractEq}) is a 1D invariant manifold and if a
manifold $\mathcal{L}$ is transversal to a vector field $J$ then the trajectory of
$\mathcal{L}$ is invariant.

Lyapunov used the {\em analyticity} of the invariant manifold for finite-dimensional
analytic vector fields $J$ to prove its existence and uniqueness near a fixed point
$\ff_0$ if $\ker m$ is a invariant subspace of the Jacobian $(DJ)_0$ of $J$ at this point
and under some ``no resonance'' conditions (the Lyapunov  auxiliary theorem
\cite{Lya1992}). Under these conditions, there exist many smooth non-analytical
manifolds, but the analytical one is unique.
\begin{theorem}[Lyapunov auxiliary theorem]\label{LyaAuxTHeo}
Let $\ker m$ have a $(DJ)_0$-invariant supplement $(\ker m)'$, $E=\ker m\oplus (\ker
m)'$. Assume that  the restriction $(DJ)_0$ onto $\ker m$ has  the spectrum $\kappa_1,
\ldots , \kappa_j$ and the restriction of this operator on the supplement $(\ker m)'$ has
the spectrum $\lambda_1, \ldots , \lambda_l$. Let the two following conditions hold:
\begin{enumerate}
\item{$0 \notin {\rm conv} \{\kappa_1, \ldots , \kappa_j\}$;}
\item{The spectra $\{\kappa_1, \ldots , \kappa_j\}$ and $\{\lambda_1, \ldots , \lambda_l\}$ are not related by any equation of the form
$$\sum_i n_i\kappa_i= \lambda_k$$ with integer $n_i$.}
\end{enumerate}
Then there exists a unique analytic solution $\ff_M$ of the invariance equation
(\ref{IM}) with condition $\ff_M=\ff_0$ for $M=m(\ff_0)$, and in a sufficiently small
vicinity of $m(\ff_0)$.
\end{theorem}
This solution is tangent to $(\ker m)'$ at point $\ff_0$.

Recently, the approach to invariant manifolds based on the invariance equation in
combination with the Lyapunov auxiliary theorem were used for the reduction of kinetic
systems \cite{Kaz1,Kaz2,Kaz3}.

\subsection{The Chapman--Enskog expansion\label{Sec:ChEnsk}}

The Chapman--Enskog and geometric singular perturbation approach assume the special
singularly perturbed structure of the equations and look for the invariant manifold in a
form of the series in the powers of a small parameter $\epsilon$. A one-parametric system
of equations is considered:
\begin{equation}\label{AbstractEqSingPert}
\partial_t f +A(f)=\frac{1}{\epsilon}Q(f).
\end{equation}
The following assumptions connect the macroscopic variables to the singular perturbation:
\begin{itemize}
\item{$m(Q(f))=0$;}
\item{for each $M\in m(U)$ the system of equations
$$Q(f)=0, \;\; m(f)=M$$
has a unique solution $\ff^{\rm eq}_M$ (in Boltzmann kinetics it is the local
Maxwellian);}
\item{$\ff^{\rm eq}_M$ is asymptotically stable and globally attracting for the fast system
$$\partial_t f =\frac{1}{\epsilon}Q(f)$$
in $(\ff^{\rm eq}_M+\ker m)\cap U$. }
\end{itemize}
Let the differential of the fast vector field $Q(f)$ at equilibrium $\ff^{\rm eq}_M$  be
$\mathcal{Q}_M$. For the Chapman--Enskog method it is important that $\mathcal{Q}_M$ is
invertible in $\ker m$. For the classical kinetic equations this assumption can be
checked using the symmetry of $\mathcal{Q}_M$  with respect to the {\em entropic inner
product} (Onsager's reciprocal relations).

The invariance equation for the singularly perturbed system (\ref{AbstractEqSingPert}) with the moment parametrization $m$
is:
\begin{equation}\label{AbstractEqSingIM}
\boxed{
\frac{1}{\epsilon}Q(\ff_M)= A(\ff_M)-(D_M \ff_M) (m(A(\ff_M))).
}
\end{equation}
The fast vector field vanishes on the right hand side of this equation because
$m(Q(\ff))=0$. The self-consistency condition $m(\ff_M)=M$ gives $$m(D_M
\ff_M)m(J)=m(J)$$ for all $J$, hence,
\begin{equation}\label{self-consistency}
m[A(\ff_M)-(D_M \ff_M) m(A(\ff_M))]=0.
\end{equation}
If we find an approximate solution of (\ref{AbstractEqSingIM}) then the corresponding
macroscopic (hydrodynamic) equation (\ref{projectedMIM}) is

\begin{equation}\label{projectedMIMSing}
\partial_t M+m(A(\ff_M))=0 .
\end{equation}

Let us represent all the operators in (\ref{AbstractEqSingIM}) by the Taylor series (for
the Boltzmann equation $A$ is the linear free flight operator, $A=v\cdot \nabla$,    and
$Q$ is the quadratic collision operator).  We look for the invariant manifold in the form
of the power series:
\begin{equation}\label{ChapmanEnskogGeneral}
\ff_M=\ff^{\rm eq}_M +\sum_{i=1}^{\infty} \epsilon^i \ff^{(i)}_M
\end{equation}
with the self-consistency condition $m(\ff_M)=M$, which implies $m(\ff^{\rm eq}_M)=M$,
$m(\ff^{(i)}_M )=0$ for $i\geq 1$. After matching the coefficients of the series in
(\ref{AbstractEqSingIM}), we obtain for every $\ff^{(i)}_M$ a linear equation
\begin{equation}\label{ChEeq}
\mathcal{Q}_M \ff^{(i)}_M = P^{(i)}(\ff^{\rm eq}_M,\ff^{(1)}_M, \ldots , \ff^{(i-1)}_M),
\end{equation}
where the polynomial operator $ P^{(i)}$ at each order $i$ can be obtained by
straightforward calculations from (\ref{AbstractEqSingIM}). Due to the self-consistency,
$m(P^{(i)})=0$ for all $i$ and the equation (\ref{ChEeq}) is solvable. The first term of
the Chapman--Enskog expansion has a simple form
\begin{equation}\label{ChE1}
\boxed{
\ff^{(1)}_M=\mathcal{Q}_M^{-1} (1-(D_M \ff^{\rm eq}_M)m) (A(\ff^{\rm eq}_M)) .
}
\end{equation}
A detailed analysis of explicit versions of this formula for the Boltzmann equation and
other kinetic equations is presented in many books and papers
\cite{Chapman,Hirschfelder}. Most of the physical applications of kinetic theory, from
the transport processes in gases to modern numerical methods (lattice Boltzmann models
\cite{LB2}) give examples of the practical applications and deciphering of this formula.
For the Boltzmann kinetics, the zero-order approximation, $\ff^{(0)}_M\approx\ff^{\rm
eq}_M$ produces  in projection on the hydrodynamic fields (\ref{projectedMIMSing}) the
compressible Euler equation. The first-order approximate invariant manifold,
$\ff^{(1)}_M\approx\ff^{\rm eq}_M+\epsilon \ff^{(1)}_M$,  gives  the compressible
Navier-Stokes equation and provides the explicit dependence of the transport coefficients
from the collision model. This bridge from the ``atomistic view to the laws of motion of
continua'' is, in some sense, the main result of the Boltzmann kinetics and follows
precisely Hilbert's request but not as rigorously as it is desired.

The calculation of higher order terms needs nothing but differentiation and calculation
of the inverse operator $\mathcal{Q}_M^{-1}$.  (Nevertheless these calculations may be
very bulky and one of the creators of the method, S. Chapman, compared reading his book
\cite{Chapman} to ``chewing glass'', cited by \cite{brush1976}). Differentiability is
needed also because the transport operator $A$ should be bounded to provide  strong sense
to the manipulations (see the discussion in \cite{Saint-RaymondCompanion}). The second
order in $\epsilon$ hydrodynamic equations (\ref{projectedMIM}) are called Burnett
equations (with $\epsilon^2$ terms) and super-Burnett equations for higher orders.

\subsection{Euler, Navier--Stokes, Burnett, and super--Burnett terms for a simple kinetic equation\label{Sec:SimplestEuler...SuperBur}}

Let us illustrate the basic construction on the simplest example (\ref{Grad101}).
\begin{eqnarray*}
\ff=\left(
\begin{array}{c} p(x)\\u(x)\\ \sigma(x) \end{array}
\right), \;
m=\left(
\begin{array}{ccc} 1 & 0 & 0 \\ 0 & 1 & 0 \end{array}
\right), \;
M=\left(
\begin{array}{c} p(x)\\u(x) \end{array}
\right), \;
\ker m=\left\{\left(\begin{array}{c} 0\\ 0\\ y \end{array}
\right)\right\},
\end{eqnarray*}
\begin{eqnarray*}
A(\ff)=\left(
\begin{array}{c} \frac{5}{3}\partial_x u \\ \partial_x p +\partial_x \sigma \\ \frac{4}{3}\partial_x u \end{array}
\right),\;
Q(\ff)=\left(
\begin{array}{c} 0\\ 0\\ -\sigma \end{array}
\right),  \;
\mathcal{Q}_M^{-1}=\mathcal{Q}_M=-1 \mbox{ on } \ker m, \\
\end{eqnarray*}
\begin{eqnarray*}
\ff^{\rm eq}_M=\left(
\begin{array}{c} p(x)\\u(x)\\ 0 \end{array}
\right), \;
D_M \ff^{\rm eq}_M=
\left(
\begin{array}{cc} 1 & 0  \\ 0 & 1  \\ 0 & 0  \end{array}
\right), \;
\ff^{(1)}_M=\left(
\begin{array}{c} 0 \\ 0 \\ -\frac{4}{3}\partial_x u \end{array}
\right).\;
\end{eqnarray*}
We hasten to remark that (1.6) is a simple linear system and can be integrated
immediately in explicit form. However, that solution contains both the fast and slow
components and it does not readily reveal the slow hydrodynamic manifold of the system.
Instead, we are interested in extracting this slow manifold by a direct method. The
Chapman-Enskog expansion is thus the tool for this extracting which we shall address
first.

The projected equations in the zeroth (Euler) and the first (Navier--Stokes) order  of
$\epsilon$ are
\begin{equation*}
\label{Grad101Eu}
\mbox{ (Euler) } \begin{array}{ll}\partial_t p = -\frac{5}{3}\partial_x u,\\
\partial_t u =-\partial_x p;
\end{array} \;\;\;
\mbox{ (Navier-Stokes) } \begin{array}{ll}
\partial_t p =-\frac{5}{3}\partial_x u,\\
\partial_t u =-\partial_x p +\epsilon \frac{4}{3}\partial_x^2 u.
\end{array}
\end{equation*}
It is straightforward to calculate the two next terms (Burnett and super-Burnett ones) but
let us introduce convenient notations to represent the whole Chapman-Enskog series for
(\ref{Grad101}). Only the third component of the invariance equation
(\ref{AbstractEqSingIM}) for (\ref{Grad101}) is non-trivial because of self-consistency
condition (\ref{self-consistency}). and we can write
\begin{equation}\label{Drad101IM}
-\frac{1}{\epsilon}\sigma_{(p,u)}= \frac{4}{3}\partial_x u-\frac{5}{3}(D_p \sigma_{(p,u)})(\partial_x u)-(D_u \sigma_{(p,u)}) (\partial_x p +\partial_x \sigma_{(p,u)}).
\end{equation}
Here, $M={(p,u)}$ and the differentials are calculated by the elementary rule: if a
function $\Phi$ depends on values of $p(x)$ and its derivatives, $\Phi=\Phi(p,\partial_x
p,\partial^2_x p,\ldots)$   then  $D_p \Phi$ is a differential operator,
$$D_p \Phi = \frac{\partial \Phi}{\partial p}+\frac{\partial \Phi}{\partial (\partial_x p)}\partial_x+\frac{\partial \Phi}{\partial (\partial_x^2 p)}\partial^2_x+\ldots$$

The equilibrium of the fast system (the Euler approximation) is known,
$\sigma_{(p,u)}^{(0)}=0$. We have already found
$\sigma_{(p,u)}^{(1)}=-\frac{4}{3}\partial_x u$ (the Navier--Stokes approximation). In
each order of the Chapman--Enskog expansion $i\geq 1$ we  get from (\ref{Drad101IM}):
\begin{equation}\label{Grad101CE}
\sigma_{(p,u)}^{(i+1)}= \frac{5}{3}(D_p \sigma_{(p,u)}^{(i)})(\partial_x u)+(D_u \sigma_{(p,u)}^{(i)})(\partial_x p)
+ \sum_{j+l=i} (D_u \sigma_{(p,u)}^{(j)})(\partial_x \sigma_{(p,u)}^{(l)})
\end{equation}

This chain of equations is nonlinear but every $\sigma_{(p,u)}^{(i+1)}$ is a linear
function of derivatives of $u$ and $p$ with constant coefficients because this sequence
starts from $ -\frac{4}{3}\partial_x u $  and the induction step in $i$ is obvious. Let
$\sigma_{(p,u)}^{(i)}$ be a linear function of derivatives of $u$ and $p$ with constant
coefficients. Then its differentials $D_p \sigma_{(p,u)}^{(i)}$ and $D_u
\sigma_{(p,u)}^{(i)}$ are linear differential operators with constant coefficients and
all terms in (\ref{Grad101CE}) are again  linear functions of derivatives of $u$ and $p$
with constant coefficients.

For $\sigma_{(p,u)}^{(2)}$ ($i+1=2$) the operators in the right hand part of
(\ref{Grad101CE}) are: $(D_p \sigma_{(p,u)}^{(1)})\allowbreak =0$,
$(D_u\sigma_{(p,u)}^{(1)})=-\frac{4}{3}\partial_x$, and in the third term in each summand
either $l=0$, $j=1$ or $l=1$, $j=0$. Therefore, for the Burnett term,
$$\sigma_{(p,u)}^{(2)}=-\frac{4}{3}\partial^2_x p .$$
For the super Burnett term in $\sigma_{(p,u)}^{(3)}$ ($i+1=3$) the operators in the right
hand part of (\ref{Grad101CE}) are $(D_p \sigma_{(p,u)}^{(2)})=-\frac{4}{3}\partial^2_x$,
$(D_u\sigma_{(p,u)}^{(2)})=0$ and in the third term, only summand with $l=j=1$ may take
non-zero value:
$$(D_u \sigma_{(p,u)}^{(1)})(\partial_x \sigma_{(p,u)}^{(1)})=(-\frac{4}{3}\partial^2_x)
(-\frac{4}{3}\partial_x u)=\frac{16}{9}\partial^3_x u .$$ Finally, $\sigma_{(p,u)}^{(3)}= -\frac{4}{9}\partial^3_x u$ and the projected equations have the form
\begin{eqnarray}
\label{Grad101Burn}
&\begin{array}{ll}\partial_t p = -\frac{5}{3}\partial_x u,\\
\partial_t u =-\partial_x p +\epsilon \frac{4}{3}\partial_x^2 u + \epsilon^2 \frac{4}{3}\partial_x^3 p
\end{array} \mbox{ (Burnett). } \\
&\begin{array}{ll}
\partial_t p =-\frac{5}{3}\partial_x u,\\
\partial_t u =-\partial_x p +\epsilon \frac{4}{3}\partial_x^2 u + \epsilon^2 \frac{4}{3}\partial_x^3 p + \epsilon^3 \frac{4}{9}\partial_x^4 u
\end{array} \mbox{ (super Burnett). }\label{Grad101supBurn}
\end{eqnarray}

To see the properties of the resulting equations, we compute the dispersion relation for
the hydrodynamic modes. Using a new space-time scale, $x^{\prime}=\epsilon^{-1} x$, and
$t^{\prime}=\epsilon^{-1} t$, and representing $u=u_k \varphi(x^{\prime},t^{\prime})$,
and $p=p_k \varphi(x^{\prime},t^{\prime})$, where
$\varphi(x^{\prime},t^{\prime})=\exp(\omega t^{\prime}+ikx^{\prime})$, and $k$ is a
real-valued wave vector, we obtain the following dispersion relations $\omega(k)$ from
the condition of a non-trivial solvability of the corresponding linear system with
respect to $u_k$ and $p_k $:
\begin{equation}
\label{dispersionns101}
\omega_{\pm}=-\frac{2}{3}k^2 \pm \frac{1}{3}i|k|\sqrt{15-4k^2},
\end{equation}
for the Navier--Stokes approximation,

\begin{equation}
\label{dispersionburnett101}
\omega_{\pm}=-\frac{2}{3}k^2 \pm \frac{1}{3}i|k|\sqrt{15+16 k^2},
\end{equation}
for the Burnett approximation (\ref{Grad101Burn}), and
\begin{equation}
\label{dispersionsburnett101}
\omega_{\pm}=\frac{2}{9}k^2 (k^2 -3) \pm \frac{1}{9}i|k|\sqrt{135+144k^2+24k^4- 4k^6 },
\end{equation}
for the super-Burnett approximation (\ref{Grad101supBurn}).

\begin{figure}
\begin{centering}\boxed{
\includegraphics[width=0.8\textwidth]{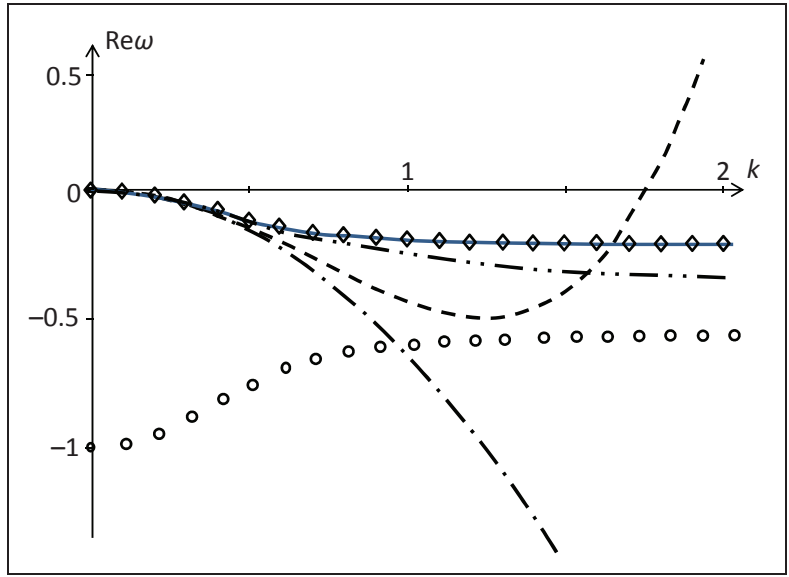}}
\caption{\label{Attenuation} Attenuation rates \cite{KGAnPh2002}. Solid: exact summation; diamonds: hydrodynamic
modes of the kinetic equations with $\epsilon=1$ (\ref{Grad101})
(they match the solid line per construction); circles: the non-hydrodynamic mode of (\ref{Grad101}), $\epsilon=1$; dash dot line: the Navier--Stokes
approximation; dash: the super--Burnett approximation; dash double dot line: the first Newton's iteration (\ref{N1E101}).
The result for the second iteration (\ref{1Nreg1101}) is indistinguishable from the exact solution at this scale.}
\end{centering}
\end{figure}

These examples demonstrate that the real part is non-positive,
${\rm Re}( \omega_{\pm}(k))\le 0$ (Fig.~\ref{Attenuation}), for the Navier--Stokes (\ref{dispersionns101}) and for
the Burnett (\ref{dispersionburnett101}) approximations, for all wave vectors. Thus,
these approximations describe attenuating acoustic waves. However, for the super-Burnett
approximation, the function ${\rm Re}( \omega_{\pm}(k))$ (\ref{dispersionsburnett101})
becomes positive as soon as $|k|>\sqrt{3}$. The equilibrium is stable within the
Navier--Stokes and the Burnett approximation, and it becomes  unstable within the
super-Burnett approximation for sufficiently short waves. Similar to the case of the
Bobylev instability of the Burnett hydrodynamics for the Boltzmann equation, the latter
result contradicts the dissipative properties of the Grad system (\ref{Grad101}): the
spectrum of the kinetic system (\ref{Grad101}) is stable for arbitrary $k$ (see
Fig.~\ref{Attenuation}). For the 13-moment system
(\ref{balanceequations})-(\ref{Grad133e}) the instability of short waves appears already
in the Burnett approximation \cite{GorKarLNP2005,KGAnPh2002} (see section \ref{Sec:Exact}
below). For the Boltzman equation this effect was discovered by Bobylev \cite{Bob}. In
Fig.~\ref{Attenuation}, we also represent the attenuation rates of the hydrodynamic and
non-hydrodynamic mode of the kinetic equations (\ref{Grad101}). The
characteristic equation of these kinetic equations reads:
\begin{equation}
\label{dispersiongrad101}
3\omega^3 +3\omega^2 + 9k^2 \omega +5k^2 =0.
\end{equation}
The two complex-conjugate roots of this equation correspond to the hydrodynamic modes,
while for the non-hydrodynamic real mode, $\omega_{nh}(k)$, $\omega_{nh}(0)=-1$, and
$\omega_{nh}\rightarrow -0.5$ as $|k|\rightarrow\infty$. The non-hydrodynamic modes of
the Grad equations are characterized by the common property that for them $\omega(0)\ne
0$. These modes  are irrelevant to the Chapman--Enskog branch of the invariant manifold.

Thus, the Chapman--Enskog expansion:
\begin{itemize}
\item{Gives  excellent, but already known on phenomenological grounds, zero and first
    order approximations -- the Euler and Navier--Stokes equations;}
\item{Provides a bridge from microscopic models of collisions to macroscopic
    transport coefficients in the known continuum equations;}
\item{Already the next correction, not known phenomenologically and hence of
    interest, does not exist because of non-physical behavior.}
\end{itemize}

The first term of the Chapman--Enskog expansion gives the possibility to evaluate the
coefficients in the phenomenological equations (like viscosity, thermal conductivity and
diffusion coefficient) from the microscopic models of collisions. The success of the
first order approximation (\ref{ChE1}) is compatible with the failure of the higher order
terms. The Burnett and Super-Burnett equations have non-physical properties, negative
viscosity for high gradients and instability for short waves. The Chapman--Enskog
expansion has to be truncated after the first order term or not truncated at all.

Such a situation when the first approximations are useful but the higher terms become
senseless is not very novel. There are at least three famous examples:
\begin{itemize}
 \item{The ``ultraviolet catastrophe'' in higher order terms because of physical
     phenomena at very short distances \cite{Dyson1952} and the perturbation series
     divergencies \cite{Zinn-JustinDFTandSF2002} are well known in quantum field
     theory, and many approaches have been developed to deal with these singularities
     \cite{Weinberg1995};}
 \item{Singularities and divergence in  the semiclassical Wentzel-Kramers-Brillouin
     (WKB) approach \cite{WKBFroman1965,WKBVoros1983,WKBHyou2004};}
 \item{The small denominators affect the convergence of the Poincar\'e series in the
     classical many body problem and the theory of nearly integrable systems. They
     may even make the perturbation series approach senseless \cite{KAMAr2}. }
\end{itemize}

Many ideas have been proposed and implemented to deal with these singularities: use of
the direct iteration method instead of power series in KAM \cite{KAM,KAM1,KAMAr2},
renormalization
 \cite{FisherRenorm1974,CollinsRenorm1984,MorrisRenorm1994}, summation and partial summation and rational approximation
of the perturbation series \cite{Mattuck1976,EllisPadeSum} and
string theories \cite{VenezString1968,DeligneString1999} in quantum field theory \cite{Weinberg1995}. Various ad hoc analytical and
numerical regularization tricks have been proposed too. Exactly solvable models give the
possibility of exhaustive analysis of the solutions. Even in the situation when they
are not applicable directly to reality we can use them as benchmarks for all perturbation
and approximation methods and for regularization tricks.

We follow this stream of ideas with the modifications required for kinetic theory. In the
next section we describe algebraic invariant manifolds for the kinetic equations
(\ref{balanceequations})-(\ref{Grad133e}), (\ref{Grad103}), (\ref{Grad101}) and
demonstrate the exact summation approach for the Chapman--Enskog series for these models.
We use these models to demonstrate the application of the Newton method to the invariance
equation (\ref{AbstractEqSingIM}).

\section{Algebraic hydrodynamic invariant manifolds and
exact summation of the Chapman--Enskog series for the simplest kinetic model
\label{Sec:Exact}}

\subsection{Grin of the vanishing cat: $\epsilon$=1\label{Sec:e=1Cat}}

At the end of the previous section we introduced a new space-time scale,
$x^{\prime}=\epsilon^{-1} x$, and $t^{\prime}=\epsilon^{-1} t$. The rescaled equations do
not depend on $\epsilon$ at all and are, at the same time, equivalent to the original
systems. Therefore, the presence of the small parameter in the equations is virtual.
``Putting $\epsilon$ back $=1$, you hope that everything will converge and single out a
nice submanifold'' \cite{McKean1965}.

In this section, we find the invariant manifold for the equations with  $\epsilon$=1.
Now, there is no fast--slow decomposition of motion. It is natural to ask: what is the
remainder of the qualitative picture of slow invariant manifold presented in
Fig.~\ref{fig1SLw}? Or an even sharper question: what we are looking for?

The rest of the fast-slow decomposition is the zeroth term in the Chapman--Enskog
expansion (\ref{ChapmanEnskogGeneral}). It starts from the equilibrium of the fast
motion, $\ff^{\rm eq}_M $. This (locally) equilibrium manifold corresponds to the limit
$\epsilon=0$. The first terms of the series for $\sigma$ for (\ref{Grad101}),
\begin{equation}\label{sigmaSimpleFirst}
\sigma=-\epsilon \frac{4}{3}\partial_x u - \epsilon^2 \frac{4}{3}\partial_x^2 p -
\epsilon^3 \frac{4}{9}\partial_x^3 u+\ldots,
\end{equation}
also bear the offprint of the zeroth approximation, $\sigma^{(0)}=0$, even when we take
$\epsilon=1$. The Chapman--Enskog procedure derives recurrently terms of the series from
the starting term, $\ff^{\rm eq}_M $.

The problem of the invariant manifold includes two difficulties: (i) it is difficult to
find any global solution or even prove its existence and (ii) there often  exists too
many different local solutions. The auxiliary Lyapunov theorem gives the first solution
of the problem near an equilibrium and several seminal hints for the further attempts.
One of them is: use the analyticity as a selection criterion. The Chapman--Enskog method
demonstrates that the inclusion of the system in the one-parametric family (parameterized
by $\epsilon$) and the requirement of analyticity up to the limit $\epsilon=0$ allows us
to select a sensible solution to the invariance equation. Even if we return to a single
system with $\epsilon=1$, the structure of the constructed invariant manifold remembers
the limit case  $\epsilon=0$... This can be considered as a manifestation of the effect
of ``the grin of the vanishing cat'': `I've often seen a cat without a grin,' thought
Alice: `but a grin without a cat! It's the most curious thing I ever saw in my life!'
(Lewis Carroll, Alice's Adventures in Wonderland.)  The small parameter disappears (we
take $\epsilon=1$) but the effect of its presence persists in the analytic invariant
manifold. There are some other effects of such a grin in kinetics \cite{GorbanYab2011}.

The use of the term ``slow manifold'' for the case $\epsilon=1$ seems to be an abuse of
language. Nevertheless, this manifold has some offprints of slowness, at least for smooth
solutions bounded by small number. The definition of slow manifolds for a single system
may be a non-trivial task \cite{Debush,GorKarLNP2005}. There is a problem with a local
definition because for a given vector field the ``slowness'' of a submanifold
cannot be invariant with respect to diffeomorphisms in a vicinity of a regular point.
Therefore we use the term ``hydrodynamic invariant manifold''.

\subsection{The pseudodifferential form of the stress tensor\label{Sec:PseodDiffStress}}

Let us return to the simplest kinetic equation (\ref{Grad101}). In order to construct the exact
solution, we first analyze the global structure of the Chapman--Enskog series given by
the recurrence formula (\ref{Grad101CE}). The first three terms (\ref{sigmaSimpleFirst})
give us a hint: the terms in the series alternate. For odd $i=1,3,\ldots$ they are
proportional to $\partial_x^i u$ and for even $i=2,4,\ldots$ they are proportional to
$\partial_x^i p$. Indeed, this structure can be proved by induction in $i$ starting in
(\ref{Grad101CE}) from the first term $-\frac{4}{3}\partial_x u$. It is sufficient
to notice that  $(D_p  \partial^{(i)}_x p  )=\partial_x^{(i)}$, $(D_p \partial^i_x u
)=0$, $(D_u \partial^i_x p )=0$, $(D_u \partial^i_x u )=\partial_x^{(i)}$ and to use the
induction assumption in (\ref{Grad101CE}).

The global structure of the Chapman--Enskog series gives the following representation of the stress
$\sigma$ on the hydrodynamic invariant manifold
\begin{equation}\label{Structure1D10M}
\sigma(x)=A(-\partial_x^2)\partial_x u(x) +B(-\partial_x^2) \partial_x^2 p(x),
\end{equation}
where $A(y)$, $B(y)$ are yet unknown functions and the sign `$-$' in the arguments is
adopted for simplicity of formulas in the Fourier transform.

It is easy to prove the structure (\ref{Structure1D10M}) without any calculation or
induction. Let us use the symmetry property of the kinetic equation (\ref{Grad101}): it
is invariant with respect to the transformation $x\mapsto -x$, $u \mapsto -u$, $p \mapsto
p$ and $\sigma \mapsto \sigma$ which transforms solutions into solutions. The invariance
equation inherits this property, the initial equilibrium ($\sigma=0$) is also symmetric
and, therefore, the expression for $\sigma(x)$ should be even. This is exactly
(\ref{Structure1D10M}) where $A(y)$ and $B(y)$ are arbitrary even functions. (If they
are, say, twice differentiable at the origin then we can represent them as functions of
$y^2$).

\subsection{The energy formula and `capillarity' of ideal gas\label{Sec:caplillarity}}

Traditionally, $\sigma$ is considered as a viscous stress tensor but the second term,
$B(-\partial_x^2) \partial_x^2 p(x)$, is proportional to second derivative of $p(x)$ and
it does not meet usual expectations ($\sigma \sim \nabla u$). Slemrod
\cite{SlemQuaterly2012,SlemCAMWA2013} noticed that the proper interpretation of  this
term is the capillarity tension rather than viscosity. This is made clear by inspection
of the energy balance formula. Let us derive  the Slemrod energy formula for the simple
model (\ref{Grad101}). The time derivative of the kinetic energy due to the first two
equations (\ref{Grad101}) is
\begin{equation}
\begin{split}
\frac{1}{2}  \partial_t \int_{-\infty}^{\infty} u^2 \,\D x &= \int_{-\infty}^{\infty} u \partial_t u \, \D x
=-\int_{-\infty}^{\infty} u\partial_x p \, \D x- \int_{-\infty}^{\infty} u\partial_x \sigma \, \D x \\ &=
- \frac{1}{2} \partial_t \frac{3}{5} \int_{-\infty}^{\infty} p^2 \, \D x + \int_{-\infty}^{\infty}  \sigma \partial_x u\, \D x
\end{split}
\end{equation}
Here we used integration by parts and assumed that all the fields with their derivatives
tend to $0$ when $x \to \pm \infty$.

In $x$-space the energy formula is
\begin{equation}\label{energy101x}
\frac{1}{2} \partial_t \left(\frac{3}{5}\int_{-\infty}^{\infty} p^2 \, \D x+\int_{-\infty}^{\infty} u^2 \,\D x\right)=
\int_{-\infty}^{\infty} \sigma  \partial_x  u\, \D x
\end{equation}
This form of the energy equation is standard. Note that the usual factor $\rho$ in front of
$u^2$ is absent because we work with the linearized equations.

Let us use in (\ref{energy101x}) the representation (\ref{Structure1D10M}) for $\sigma$
and notice that $\partial_x   u= -\frac{3}{5}\partial_t p$:
$$\int_{-\infty}^{\infty} \sigma  \partial_x   u\, \D x
=\int_{-\infty}^{\infty} A(-\partial_x^2)\partial^2_x u \, \D x
-\frac{3}{5}\int_{-\infty}^{\infty} (\partial_t p) [ B(-\partial_x^2) \partial_x^2 p ] \,
\D x$$  The operator $B(-\partial_x^2) \partial_x^2$ is symmetric, therefore,
$$\int_{-\infty}^{\infty} (\partial_t p) [ B(-\partial_x^2) \partial_x^2 p ] \, \D
x=\frac{1}{2}\partial_t\left(\int_{-\infty}^{\infty} p [ B(-\partial_x^2) \partial_x^2 p
] \, \D x\right)$$ The quadratic form,
\begin{equation}\label{CapillarEnergyX}
U_c=\frac{3}{5}\int_{-\infty}^{\infty} p ( B(-\partial_x^2)
\partial_x^2 p ) \, \D x= - \frac{3}{5}\int_{-\infty}^{\infty} (\partial_x p )( B(-\partial_x^2)
\partial_x p )\, \D x
\end{equation}
may be considered as a part of the energy. Moreover, if the function $B(y)$ is negative then
this form is  positive. Due to Parseval's identity we have
\begin{equation}\label{CapillarEnergyK}
U_c=-\frac{3}{5}\int_{-\infty}^{\infty} k^2 B(k^2) |p_k|^2 \, \D k.
\end{equation}
Finally, the energy formula in $x$-space is
\begin{equation}\label{EnergyXfin}
\boxed{\begin{aligned}
 \frac{1}{2} \partial_t \int_{-\infty}^{\infty} \left(\frac{3}{5} p^2+u^2 -
\frac{3}{5} (\partial_x p )( B(-\partial_x^2)
\partial_x p ) \right) \, \D x \\
=\int_{-\infty}^{\infty} (\partial_x u)(A(-\partial_x^2)\partial_x u) \, \D x
\end{aligned}}
\end{equation}
In $k$-space it has the form
\begin{equation}\label{energy101k}
\begin{split}
\frac{1}{2} \partial_t \int_{-\infty}^{\infty}  \left(\frac{3}{5}|p_k|^2+|u_k|^2
-\frac{3}{5} k^2
B(k^2)|p_k|^2 \right) \, \D k =\int_{-\infty}^{\infty} k^2A(k^2)|u_k|^2 \, \D k
\end{split}
\end{equation}
It is worth mentioning that the functions $A(k^2)$ and $B(k^2)$ are negative (see
Sec.~\ref{IEFourier}). If we keep only the first non-trivial terms, $A=B=-\frac{4}{3}$,
then the energy for\-mula becomes
\begin{eqnarray}\label{EnergyXKlinfin}
&& \frac{1}{2}\partial_t \int_{-\infty}^{\infty}\left(\frac{3}{5} p^2 + u^2
+\frac{4}{5}(\partial_x p)^2 \right) \, \D x
=-\frac{4}{3}\int_{-\infty}^{\infty} (\partial_x u)^2 \, \D x ;\\
&&\frac{1}{2}\partial_t \int_{-\infty}^{\infty} \left(\frac{3}{5} |p_k|^2 \, \D k+ |u_k|^2
+\frac{4}{5} k^2 |p_k|^2\right)  \, \D k
=-\frac{4}{3}\int_{-\infty}^{\infty} k^2|u_k|^2 \, \D k .
\end{eqnarray}

Slemrod represents the structure of the obtained energy formula as
\begin{equation}\label{CapillIdea}
\begin{split}
\partial_t ({\rm MECHANICAL\ ENERGY})+\partial_t
({\rm CAPILLARITY\ ENERGY})\\ = {\rm VISCOUS\
DISSIPATION}.
\end{split}
\end{equation}
The capillarity terms $(\partial_x p)^2$ in the energy density are standard in the
thermodynamics of phase transitions.

The bulk capillarity terms in fluid mechanics were introduced into the Navier--Stokes
equations by Korteweg \cite{Korteweg1901} (for a review of some further results see
\cite{Dunn1985}). Such terms appear naturally in theories of the phase transitions such
as van der Waals liquids \cite{SlemrodVdW},  Ginzburg--Landau \cite{AransonGinLand2002}
and Cahn--Hilliard equations \cite{CahnHilliard1958,Cahn1959}, and phase fields models
\cite{ChenPhaseField2002}. Surprisingly, such terms are also found in the ideal gas
dynamics as a consequence of the Chapman--Enskog expansion \cite{Slem1,Slem2}. In
higher-order approximations, the viscosity is reduced by the terms which are similar to
Korteweg's capillarity. Finally, in the energy formula for the exact sum of the
Chapman--Enskog expansion we see terms of the same form: the viscous dissipation is
decreased and the additional term  appears in the energy (\ref{EnergyXfin}),
(\ref{energy101k}).

\subsection{Algebraic invariant manifold in Fourier representation\label{IEFourier}}

It is convenient to work with the pseudodifferential operators like
(\ref{Structure1D10M}) in Fourier space. Let us denote $p_k$, $u_k$ and
$\sigma_k$, where $k$ is the `wave vector' (space frequency).

The Fourier-transformed kinetic equation (\ref{Grad101}) takes the form
($\epsilon=1$):
\begin{equation}\label{Grad101F}
\begin{split}
\partial_t p_k &=-\frac{5}{3}ik u_k ,\\
\partial_t u_k &=-ik p_k -ik \sigma_k ,\\
\partial_t \sigma_k &=-\frac{4}{3}ik u_k
-\sigma_k .
\end{split}
\end{equation}

We know already that the result of the reduction should be a function $\sigma_k (u_k ,
p_k ,k)$ of the following form:
\begin{equation}\label{Parametriz1D10M}
\sigma_{k}(u_k , p_k ,k) =ikA(k^2)u_k -k^2 B(k^2)p_k ,
\end{equation}
where $A$ and $B$ are unknown real-valued functions of $k^2$.

The question of the summation of the Chapman--Enskog series amounts to finding the two
functions, $A(k^2 )$ and $B(k^2 )$. Let us write the invariance equation for unknown
functions $A$ and $B$. We can compute the time derivative of $\sigma_{k}(u_k , p_k ,k)$
in { two } different ways. First, we use the right hand side of the third equation in
(\ref{Grad101F}). We find the {\em microscopic} time derivative:
\begin{equation}
\label{101microderivat}
 \partial_t^{\rm micro}\sigma_k =-ik\left(\frac{4}{3}+A\right)u_k
+k^2 Bp_k .
\end{equation}

Secondly, let us use chain rule
and the first two equations in (\ref{Grad101F}). We find
 the {\em macroscopic} time derivative:
\begin{equation}
\begin{split}\label{101macroderivat}
\partial_t^{\rm macro} \sigma_k & =\frac{\partial
\sigma_k }{\partial u_k }\partial_t u_k +\frac{\partial \sigma_k
}{\partial p_k }\partial_t p_k \\  &=ikA\left(-ikp_k
-ik\sigma_k \right)-k^2 B\left(-\frac{5}{3}iku_k \right)\\
&=ik\left(\frac{5}{3}k^2 B+k^2 A\right)u_k +k^2 \left(A-k^2
B\right)p_k .
\end{split}
\end{equation}

The microscopic time derivative should coincide with the
macroscopic time derivative for all values of $u_k$ and $p_k$.
This is the invariance equation:
\begin{equation}\label{FourierInvEq1D10M}
\partial_t^{\rm macro} \sigma_k =\partial_t^{\rm
micro}\sigma_k.
\end{equation}
 For the kinetic system (\ref{Grad101F}), it reduces to a system of two
quadratic equations for functions $A(k^2)$ and $B(k^2)$:
\begin{equation}\label{AlgebraicIM1D10M}
\begin{split}
F(A,B,k)&=-A-\frac{4}{3}-k^2 \left(\frac{5}{3}B+A^2
\right)=0,\\  G(A,B,k)&=-B+A\left(1-k^2 B\right)=0.
\end{split}
\end{equation}

The Taylor series for $A(k^2)$, $B(k^2)$ correspond exactly to the Chapman--Enskog
series: if we look for these functions in the form $A(y)=\sum_{l\geq 0}a_l y^l$ and
$B(y)=\sum_{l\geq 0}b_l y^l$ then from (\ref{AlgebraicIM1D10M}) we find immediately
$a_0=b_0=-\frac{4}{3}$ (these are exactly the Navier--Stokes and Burnett terms) and the
recurrence equation for $a_{i+1}$, $b_{i+1}$:
\begin{equation}
\label{recurrent101}
\begin{split}
a_{n+1}&=\frac{5}{3}b_n +\sum_{m=0}^{n}a_{n-m}a_m ,\\
b_{n+1}&=a_{n+1}+\sum_{m=0}^{n}a_{n-m}b_m .
\end{split}
\end{equation}
The initial condition for this set of equations are the Navier--Stokes and the Burnett
terms $a_0=b_0=-\frac{4}{3}$.

The Newton method for the invariance equation (\ref{AlgebraicIM1D10M}) generates the
sequence $A_i(k^2)$, $B_i(k^2)$, where the differences, $\delta A_{i+1} = A_{i+1}-A_i$
and $\delta B_{i+1} = B_{i+1}-B_i$ satisfy the system of linear equations

\begin{equation*}
\label{1Nitertation101}
\left( \begin{array}{cc}
\frac{\partial F(A,B,k^2)}{\partial A}\left.\right|_{(A_i , B_i) }
 &
\frac{\partial F(A,B,k^2)}{\partial B}\left.\right|_{(A_i , B_i)}
\\
\frac{\partial G(A,B,k^2)}{\partial A}\left.\right|_{(A_i , B_i)}
&
\frac{\partial G(A,B,k^2)}{\partial B}\left.\right|_{(A_i , B_i)}
\\
\end{array} \right)
\left(\begin{array}{c}\delta A_{i+1} \\ \delta B_{i+1} \\\end{array}\right)
+\left(\begin{array}{c}F(A_i , B_i ,k^2)\\G(A_i , B_i , k^2)\\\end{array}\right)=0.
\end{equation*}
Rewrite this system in the explicit form:
\begin{equation*}
\left(\begin{array}{cc} -(1+2k^2 A_{i} )
 &
-\frac{5}{3}k^2
\\
1-k^2 B_{i}
&
-(1+k^2 A_{i} )
\\
\end{array} \right)
\left(\begin{array}{c}\delta A_{i+1} \\ \delta B_{i+1} \\\end{array}\right)
+\left(\begin{array}{c}F(A_i , B_i ,k^2)\\G(A_i , B_i ,
k^2)\\\end{array}\right)=0.
\end{equation*}
Let us start from the zeroth-order term of the Chapman--Enskog expansion (Euler's
approximation), $A_0=B_0=0$. Then, the first Newton's iteration gives
\begin{equation}
\label{N1E101}
A_1 = B_1 = -\frac{4}{3+5k^2 }.
\end{equation}
The second Newton's iteration also gives the negative rational functions
\begin{equation}
\label{1Nreg1101}
\begin{split}
A_2 &=-\frac{4(27+63k^2 +153k^2 k^2 +125k^2 k^2 k^2 )}{3(3+5k^2 )
(9+9k^2 +67k^2 k^2 +75 k^2 k^2 k^2 )},\\
B_2 &=-\frac{4(9+33k^2 +115k^2 k^2 +75 k^2 k^2 k^2 )}
{(3+5k^2 )(9+9k^2 +67k^2 k^2 +75 k^2 k^2 k^2 )} .
\end{split}
\end{equation}
The corresponding attenuation rates are shown in Fig.~\ref{Attenuation}. They are stable
and converge fast to the exact solutions. At the infinity, $k^2 \to \infty$, the second
iteration has the same limit, as the exact solution: $k^2 A_2 \to - \frac{4}{9}$ and $k^2
B_2 \to -\frac{4}{5}$ (compare to Sec.~\ref{HighFreqAs}).

Thus, we made three steps:
\begin{enumerate}
\item{We used the invariance equation, Chapman--Enskog procedure and the
    symmetry properties to find a linear space where the hydrodynamic invariant
    manifold is located. This space is parameterized by two functions of one
    variable (\ref{Parametriz1D10M});}
\item{We used the invariance equation and defined an algebraic manifold in this
    space. For the simple kinetic system (\ref{Grad101}), (\ref{Grad101F}) this
    manifold is given by the system of two quadratic equations which depends linearly
    on $k^2$ (\ref{AlgebraicIM1D10M}).}
\item{We found that Newton's iterations for the invariant manifold demonstrate
    much better approximation properties than the truncated Chapman--Enskog.}
\end{enumerate}

\subsection{Stability of the exact hydrodynamic system and saturation of dissipation for short waves\label{stabSat}}

Stability is one of the first questions to analyze. There exists a series of simple
general statements about the preservation of stability, well-posedness and hyperbolicity
in the exact hydrodynamics. Indeed, any solution of the exact hydrodynamics is the projection
of a solution of the initial equation from the invariant manifold onto the hydrodynamic
moments (Figs.~\ref{fig1SLw}, \ref{McKeanDiag}) and the projection of a bounded solution
is bounded. (In infinite dimension we have to assume that the projection is continuous
with respect to the chosen norms.) Several statements of this type are discussed in
Sec.~\ref{Sec:ExactLinKin}. Nevertheless, a direct analysis of dispersion relations and
attenuation rates is instructive. Knowing $A(k^2)$ and $B(k^2)$, the dispersion relation
for the hydrodynamic modes can be derived:
\begin{equation}
\label{dispersion101}
\omega_{\pm}=\frac{k^2 A}{2} \pm
i\frac{|k|}{2}\sqrt{\frac{20}{3} (1-k^2 B)-k^2 A^2}.
\end{equation}
It is convenient to reduce the consideration to a single function. Solving the system
(\ref{AlgebraicIM1D10M}) for $B$, and introducing a new function,  $X(k^2 )=k^2 B(k^2 )$,
we obtain an equivalent cubic equation:
\begin{equation}
\label{factor101}
-\frac{5}{3}(X-1)^2 \left(X+\frac{4}{5}\right)=\frac{X}{k^2 }.
\end{equation}
Since the hydrodynamic manifold should be represented  by the real-valued functions
$A(k^2)$ and $B(k^2)$ (\ref{Parametriz1D10M}), we are only interested in the real-valued roots of (\ref{factor101}).

An elementary analysis gives:  the real-valued root $X(k^2 )$ of (\ref{factor101}) is
unique and { negative} for all finite values $k^2 $. Moreover, the function $X(k^2 )$ is
a monotonic function of $k^2 $. The limiting values are:
\begin{equation}
\label{limits101}
\lim_{|k|\rightarrow0}X(k^2 )=0, \quad
\lim_{|k|\rightarrow\infty}X(k^2 )=-0.8.
\end{equation}

Under the conditions just mentioned, the function under the root in (\ref{dispersion101})
is negative for all values of the wave vector $k$, including the limits, and we come to
the following dispersion law:
\begin{equation}
\label{frequency101}
\omega_{\pm}=\frac{X}{2(1-X)}\pm i\frac{|k|}{2}\sqrt{\frac{5X^2
-16X +20}{3}},
\end{equation}
where $X=X(k^2 )$ is the real-valued root of equation (\ref{factor101}). Since $0 > X(k^2
)>-1$ for all $|k|>0$, the attenuation rate, ${\rm Re}( \omega_{\pm})$, is negative for
all $|k|>0$, and the exact acoustic spectrum of the Chapman--Enskog procedure {\it is
stable for arbitrary wave lengths} (Fig.~\ref{Attenuation}, solid line). In the
short-wave limit, from (\ref{frequency101}) we obtain:
\begin{equation}
\label{limit101}
\lim_{|k|\rightarrow\infty}{\rm Re} \omega_{\pm}=-\frac{2}{9}, ;\ \;\; \lim_{|k|\rightarrow\infty} \frac{{\rm Im} \omega_{\pm}}{|k|} = \pm
\sqrt{3}.
\end{equation}

\subsection{Expansion at $k^2=\infty$ and matched asymptotics\label{HighFreqAs}}

For large values of $k^2$, a version of the Chapman--Enskog expansion at an
infinitely-distant point is useful. Let us rewrite the algebraic equation for the
invariant manifold (\ref{AlgebraicIM1D10M}) in the form
\begin{equation}\label{AlgebraicIM1D10M1/k^2}
\begin{split}
\frac{5}{3}B+A^2&=-\varsigma (\frac{4}{3}+A),\\
AB&=\varsigma(A-B),
\end{split}
\end{equation}
where $\varsigma=1/k^2$. For the analytic solutions near the point $\varsigma=0$ the
Taylor series is: $A=\sum_{l=1}^{\infty} \alpha_l \varsigma^l$, $B=\sum_{l=1}^{\infty}
\beta_l \varsigma^l$, where $\alpha_1=-\frac{4}{9}$, $\beta_1=-\frac{4}{5}$,
$\alpha_2=\frac{80}{2187}$, $\beta_2=\frac{4}{27}$, ... . The first term gives for the
frequency (\ref{dispersion101}) the same limit:
\begin{equation}
\label{dispersion101infty}
\omega_{\pm}=-\frac{2}{9} \pm i {|k|}{\sqrt{3}},
\end{equation}
and the higher terms give some corrections.

Let us match the Navier--Stokes term and the first term in the $1/k^2$ expansion. We get:
\begin{equation}\label{MergedAsymptotics}
A\approx-\frac{4}{3+9 k^2}, \;\; B\approx-\frac{4}{3+5k^2}
\end{equation}
and
\begin{equation}
\sigma_{k}=ikA(k^2)u_k -k^2 B(k^2)p_k \approx -\frac{4ik}{3+9k^2}u_k +\frac{4k^2}{3+5k^2}p_k .
\end{equation}

This simplest non-locality captures the main effects: the asymptotic for short waves
(large $k^2$) and the Navier--Stokes approximation for hydrodynamics for smooth solutions
with bounded derivatives and small Knudsen and Mach numbers (small $k^2$).

The saturation of dissipation at large $k^2$ is a universal effect and hydrodynamics that
do not take this effect into account cannot pretend to be a universal asymptotic
equation.

This section demonstrates that for the simple kinetic model (\ref{Grad101}):
\begin{itemize}
\item{The Chapman--Enskog series amounts to an algebraic invariant manifold, and the
    ``smallness" of the Knudsen number $\epsilon$ used to develop the Chap\-man--Enskog
    procedure {\it is no longer necessary}}.
\item{The exact dispersion relation (\ref{frequency101}) on the algebraic invariant
    manifold {\em is stable for all wave lengths}.}
\item{The exact result of the Chapman--Enskog procedure has a clear non-poly\-nomial
    character. The resulting exact hydrodynamics are {\it essentially  nonlocal} in
    space. For this reason, even if the hydrodynamic equations of a certain level of
    the approximation { are} stable, they cannot reproduce the non-polynomial
    behavior for sufficiently short waves.}
\item{The Newton iterations for the invariance equations provide much better results
    than the Chapman--Enskog expansion. The first iteration gives the Navier--Stokes
    asymptotic for long waves and the qualitatively correct behavior with saturation
    for short waves. The second iteration gives the proper higher order approximation
    in the long wave limit and the quantitatively proper asymptotic for short waves.}
\end{itemize}

In the next section we extend these results to a general linear kinetic equation.

\section{Algebraic invariant manifold for general linear kinetics in 1D\label{Sec:ExactLinKin}}

\subsection{General form of the invariance equation for 1D linear kinetics\label{GenLinIE}}

For linearized kinetic equations, it is convenient to start directly with the Fourier
transformed system.

Let us consider two sets of variables: macroscopic variables $M$ and microscopic
variables $\mu$. The corresponding vector spaces are $E_M$ ($M\in E_M$) and $E_{\mu}$
($\mu\in E_{\mu}$), $k$ is the wave vector and  the initial kinetic system in the Fourier
space for functions $M_k(t)$ and $\mu_k(t)$ has the following form:
\begin{equation}\label{genLinKIn}
\begin{split}
\partial_t M_k &= ik L_{MM}M_k+ik L_{M\mu}\mu_k; \\
\partial_t \mu_k&=ik L_{\mu M}M_k+ik L_{\mu \mu}\mu_k + C\mu_k,
\end{split}
\end{equation}
where $L_{MM}: E_M\to E_M$, $ L_{M\mu} E_{\mu} \to E_M $, $L_{\mu M}: E_M \to  E_{\mu} $,
$L_{\mu \mu}:  E_{\mu} \to E_{\mu}$, and $C : E_{\mu} \to  E_{\mu}$ are constant linear
operators (matrices).

The only requirement for the following algebra is: the operator $C : E_{\mu} \to E_{\mu}$
is invertible. (Of course, for further properties like stability of reduced equations we
need more assumptions like stability of the whole system (\ref{genLinKIn}) and negative definiteness of $C$.)

We look for a hydrodynamic invariant manifold in the form
\begin{equation}\label{generalAnz}
\mu_k=\mathcal{X}(k)M_k,
\end{equation}
where $\mathcal{X}(k): E_M \to E_{\mu}$ is a linear map for all $k$.

The corresponding exact hydrodynamic equation is
\begin{equation}\label{exactLinHyd}\boxed{
\partial_t M_k=ik [L_{MM}+L_{M\mu}\mathcal{X}(k)]M_k.}
\end{equation}
Calculate the micro- and macroscopic derivatives of $\mu_k$ (\ref{generalAnz}) exactly as
in (\ref{101microderivat}), (\ref{101macroderivat}):
\begin{equation}
\begin{split}\label{micrimacrogen}
\partial_t^{\rm micro} \mu_k &= [i k L_{\mu M} + i k L_{\mu \mu} \mathcal{X}(k) + C \mathcal{X}(k)]M_k ; \\
\partial_t^{\rm macro} \mu_k &= [i k \mathcal{X}(k) L_{MM} + ik \mathcal{X}(k) L_{M\mu} \mathcal{X}(k)] M_k.
\end{split}
\end{equation}
The invariance equation for $\mathcal{X}(k)$ is again a system of algebraic equations (a quadratic matrix equation):
\begin{equation}\label{generalIE}\boxed{
\mathcal{X}(k)=i k C^{-1}[- L_{\mu M} + (\mathcal{X}(k) L_{MM} - L_{\mu \mu} \mathcal{X}(k)) + \mathcal{X}(k) L_{M\mu} \mathcal{X}(k)].}
\end{equation}
This is a general invariance equation for linear kinetic systems (\ref{genLinKIn}). The
Chapman--Enskog series is a Taylor expansion for the solution of this equation at $k=0$.
Thus, immediately we get the first terms:
\begin{equation*}
\mathcal{X}(0)=0, \;  \mathcal{X}'(0)=-i C^{-1}L_{\mu M} , \;
\mathcal{X}''(0)= 2 C^{-1}(C^{-1}L_{\mu M}L_{MM} - L_{\mu \mu}C^{-1} L_{\mu M}).
\end{equation*}
The sequence of the  Euler, Navier--Stokes and Burnett approximations is:
\begin{equation}\label{LinChEnChain}
\begin{split}
\partial_t M_k=&ik L_{MM}M_k \;\; \mbox{(Euler)};\\
\partial_t M_k=&ik L_{MM}M_k+k^2 L_{M\mu}C^{-1}L_{\mu M}M_k\;\; \mbox{(Navier--Stokes)};\\
\partial_t M_k=&ik L_{MM}M_k+k^2 L_{M\mu}C^{-1}L_{\mu M}M_k\\
&+ik^3 L_{M\mu} C^{-1}(C^{-1}L_{\mu M}L_{MM} - L_{\mu \mu}C^{-1} L_{\mu M})M_k\;\;
\mbox{(Burnett)}.
\end{split}
\end{equation}
Let us use the identity $\mathcal{X}(0)=0 $ and the fact that the functions in the $x$-space are real-valued. We can separate
odd and even parts of $\mathcal{X}(k)$ and write
\begin{equation}\label{liftingGen}
\mathcal{X}(k)=ik \mathcal{A}(k^2)+k^2\mathcal{B}(k^2),
\end{equation}
where $\mathcal{A}(y)$ and $\mathcal{B}(y)$ are real-valued matrices. For these unknowns, the
invariance equation is even closer to the simple example (\ref{AlgebraicIM1D10M}):
\begin{equation}\label{generalIEreal}
\begin{split}
\mathcal{A}(k^2)= &C^{-1}[- L_{\mu M} + k^2 (\mathcal{B}(k^2) L_{MM} -
L_{\mu \mu} \mathcal{B}(k^2))\\ &- k^2 \mathcal{A}(k^2) L_{M\mu} \mathcal{A}(k^2)
+ k^4 \mathcal{B}(k^2) L_{M\mu} \mathcal{B}(k^2)],\\
\mathcal{B}(k^2)= &-C^{-1}[ (\mathcal{A}(k^2) L_{MM} - L_{\mu \mu} \mathcal{A}(k^2))\\
& +k^2 \mathcal{A}(k^2) L_{M\mu} \mathcal{B}(k^2) + k^2 \mathcal{B}(k^2) L_{M\mu}
\mathcal{A}(k^2)].
\end{split}
\end{equation}

\subsection{Hyperbolicity of exact hydrodynamics\label{Sec:Hyperbol}}

Hyperbolicity is an important property of the exact hydrodynamics. Let us recall that the
linear system represented in Fourier space by the equation
$$\partial_t u_k = -i A(k) u_k$$
is  {\em hyperbolic} if for every $t\geq 0$ the operator $\exp(-it A(k))$ is uniformly
bounded as a function of $k$ (it is sufficient to take $t=1$). This means that the
Cauchy problem for this system is well-posed {\it forward } in time.

This system is {\em strongly hyperbolic} if for every $t\in \mathbb{R}$ the operator
$\exp(-it A(k))$ is uniformly bounded as a function of $k$ (it is sufficient to take
$t=\pm 1$). This means that the Cauchy problem for this system is well-posed both {\em forward
and backward} in time.

\begin{proposition}[Preservation of hyperbolicity]
Let the original system (\ref{genLinKIn}) be (strongly) hyperbolic. Then the reduced
system (\ref{exactLinHyd}) is also (strongly) hyperbolic if the lifting operator
$\mathcal{X}(k)$ (\ref{generalAnz}) is a bounded function of $k$.
\end{proposition}
\begin{proof}
Hyperbolicity (strong hyperbolicity) is just a requirement of the  uniform boundedness in
$k$ of the solutions of (\ref{genLinKIn}) for each $t>0$ (or for all $t$) with uniformly
bounded in $k$ initial conditions. For the exact hydrodynamics, solutions of the
projected equations are projections of the solutions of the original system. Let the
original system (\ref{genLinKIn}) be (strongly) hyperbolic. If the lifting operator
$\mathcal{X}(k)$ is a bounded function of $k$ then for the uniformly bounded initial
condition $M_k$ the corresponding initial value $\mu_k=\mathcal{X}(k)M_k$ is also bounded
and, due to the hyperbolicity of (\ref{genLinKIn}), the projection of the solution is
uniformly bounded in $k$ for all $t\geq0$. In the following commutative diagram, the
upper horizontal arrow and the vertical arrows are the bounded operators, hence the lower
horizontal arrow is also a bounded operator.
\begin{equation}\label{McKeanDigr}
  \begin{CD}
(M_k(0),\mu_k(0))@> {\mbox{Time shift (initial eq.)}} >> (M_k(t), \mu_k(t))\\
@A\mbox{Lifting}AA                                        @VV\mbox{Projection}V\\
M_{k}(0)  @> {\mbox{Exact hydrodynamics}}>> M_k(t)
  \end{CD}
\end{equation}
\end{proof}

To analyze the boundedness of the lifting operator we have to study the asymptotics of
the solution of the invariance equation at the infinitely-distant point $k^2=\infty$. If
this is a regular point then  we can find the Taylor expansion in powers of
$\varsigma=\frac{1}{k^2}$, $A=\sum_l \alpha_l \varsigma^l$ and $B=\sum_l \beta_l
\varsigma^l$. For the boundedness of $\mathcal{X}(k)$ (\ref{liftingGen}) we should take
in these series $\alpha_0=\beta_0=0$. If the solution of the invariance equation is a
real analytic function for $0\geq k^2 \geq \infty$ then the condition is sufficient for
the hyperbolicity of the projected equation (\ref{exactLinHyd}). If $\mathcal{X}(k)$ is an
exact solution of the algebraic invariance equation (\ref{generalIE}) then the
hydrodynamic equation (\ref{exactLinHyd}) gives the exact reduction of
(\ref{genLinKIn}). Various approximations give the approximate reduction like the
Chapman--Enskog approximations (\ref{LinChEnChain}).

The expansion near an infinitely-distant point is useful but may be not so
straightforward. Nevertheless, if such an expansion exists then we can immediately
produce the matched asymptotics.

Thus, as we can see, the summation of the Chapman--Enskog series to an algebraic manifold
is not just a coincidence but a typical effect for kinetic equations. For a specific
kinetic system we have to make use of all the existing symmetries like parity and
rotation symmetry in order to reduce the dimension of the invariance equation and to
select the proper physical solution. Another simple but important condition is that all
the kinetic and hydrodynamic variables should be real-valued. The third selection rule is
the behavior of the spectrum near $k=0$: the attenuation rate should go to zero when
$k\to 0$.

The Chapman--Enskog expansion is a Taylor series (in $k$) for the solution of the
invariance equation. In general, there is no reason to believe that the  first few terms
of the Taylor series at $k=0$ properly describe the asymptotic behavior of the solutions
of the invariance equation (\ref{generalIE}) for all $k$. Already the simple examples
such as (\ref{Grad101F}) reveal that the exact hydrodynamic is essentially nonlocal and
the behavior of the attenuation rate at $k\to \infty$ does not correspond to {\em any}
truncation of the Chapman--Enskog series.

Of course, for a numerical solution of (\ref{generalIE}), the Taylor series expansion is
not the best approach. The Newton method gives much better results and even the first
approximation may be very close to the solution \cite{ColKarKroe2007_2}.

In the next section we show that for more complex kinetic equations the situation may be
even more involved and both the truncation and the summation of the whole series may become meaningless
for sufficiently large $k$. In these cases, the hydrodynamic solution of the
invariance equations does not exist for large $k$ and the whole problem of hydrodynamic
reduction has no solution. We will see how the hydrodynamic description is destroyed and
the coupling between hydrodynamic and non-hydrodynamic modes becomes permanent and
indestructible. Perhaps, the only advice in this situation may be to change the set of
variables or to modify the projector onto these variables: if hydrodynamics exist, then
the set of hydrodynamic variables or the projection onto these variables should be
different.

\subsection{Destruction of hydrodynamic invariant manifold for short waves in the moment equations\label{Sec:GradDestroy}}

In this section we study the one-dimensional version
of the Grad  equations (\ref{balanceequations}) and (\ref{Grad133e})
in the $k$-representation:
\begin{equation}
\label{Grad131}
\begin{split}
\partial_t \rho_k &=-iku_k ,\\
\partial_t u_k &=-ik\rho_k -ikT_k -ik\sigma_k ,\\
\partial_t T_k &=-\frac{2}{3}iku_k -\frac{2}{3}ikq_k ,\\
\partial_t   \sigma_k   &=-\frac{4}{3}iku_k   -\frac{8}{15}ikq_k
-\sigma_k ,\\
\partial_t q_k &=-\frac{5}{2}ikT_k -ik\sigma_k -\frac{2}{3}q_k .
\end{split}
\end{equation}
The Grad system (\ref{Grad131})  provides  the  simplest  coupling
of  the hydrodynamic  variables  $\rho_k  $,  $u_k   $,   and $T_k
$   to   the non-hydrodynamic variables, $\sigma_k $ and $q_k $,
the latter  is the heat flux. We need to  reduce the Grad system (\ref{Grad131}) to
the  three  hydrodynamic  equations  with respect to  the
variables  $\rho_k  $,  $u_k  $,  and  $T_k  $. That  is, in the general notations of the previous section,
$M=\rho_k, u_k, T_k $ $\mu=\sigma_k, q_k $ and    we have
to express the functions $\sigma_k $ and $q_k $ in terms of $
\rho_k  $,  $u_k   $,   and   $T_k   $:
\begin{equation*}
\sigma_k =\sigma_k (\rho_k ,u_k , T_k , k),\; \; \;
q_k =q_k (\rho_k ,u_k , T_k , k).
\end{equation*}

The derivation of the invariance equation for the system
(\ref{Grad131}) goes along the same lines as in the previous
sections. The quantities $\rho$ and $T$ are scalars, $u$ and $q$ are (1D) vectors,
and the (1D) stress `tensor' $\sigma$ is again a scalar.
The vectors and scalars transform differently under the parity transformation $x \mapsto -x$, $k \mapsto -k$.
We use this symmetry property and find the representation (\ref{generalAnz}) of $\sigma, q$ similar to (\ref{Parametriz1D10M}):
\begin{equation}\label{form131}
\begin{split}
   \sigma_{k} & =  ikA(k^2)u_{k}-k^2B(k^2)\rho_{k}-k^2C(k^2)T_{k},\\
     q_{k} & =  ikX(k^2)\rho_{k}+ikY(k^2)T_{k}-k^2Z(k^2)u_{k},
\end{split}
\end{equation}
where the functions $A,\dots, Z$ are the unknowns in the invariance equation.
By the nature of the CE recurrence procedure for the real-valued in $x$-space kinetic equations,
$A,\dots,Z$ are real-valued functions.

Let us find the microscopic and macroscopic time derivatives (\ref{micrimacrogen}). Computing the {\em microscopic} time  derivative
of the functions (\ref{form131}), due to the two last equations of the  Grad
system (\ref{Grad131}) we derive:
\begin{equation*}
\label{micro131}
\begin{split}
\partial_t^{\rm micro}\sigma_k          &=-ik\left(\frac{4}{3}-\frac{8}{15}k^2
Z+A\right)u_k \\
&+k^2 \left(\frac{8}{15}X+B\right)\rho_k
+ k^2 \left(\frac{8}{15}Y+C\right)T_k ,\\
\partial_t^{\rm micro}q_k &= k^2 \left(A+\frac{2}{3}Z\right)u_k +
ik\left(k^2   B-\frac{2}{3}X\right)\rho_k \\
& -ik \left(\frac{5}{2}-k^2   C -\frac{2}{3}Y\right)T_k .
\end{split}
\end{equation*}
On the other hand, computing the {\em macroscopic} time  derivative of the functions (\ref{form131})
due  to  the first three equations of the system (\ref{Grad131}), we obtain:
\begin{equation*}
\label{macro131sigma}
\begin{split}
\partial_t^{\rm macro}\sigma_k    &=\frac{\partial\sigma_k       }{\partial
u_k }\partial_t u_k                 +
\frac{\partial\sigma_k       }{\partial
\rho_k }\partial_t \rho +
\frac{\partial\sigma_k       }{\partial
T_k }\partial_t T_k \\
&=ik\left(k^2 A^2 +k^2 B+\frac{2}{3}k^2 C-\frac{2}{3}k^2 k^2  CZ\right)u_k
\\
&+\left(k^2 A-k^2 k^2 AB-\frac{2}{3}k^2 k^2 CX\right)\rho_k\\
&+\left(k^2 A-k^2 k^2 AC-\frac{2}{3}k^2 k^2 CY\right)T_k ;\\
\end{split}
\end{equation*}
\begin{equation*}
\label{macro131q}
\begin{split}
\partial_t^{\rm macro}q_k       &=\frac{\partial q_k       }{\partial
u_k }\partial_t u_k                 +
\frac{\partial q_k       }{\partial
\rho_k }\partial_t \rho u_k +
\frac{\partial q_k       }{\partial
T_k }\partial_t T_k \\
&=\left(-k^2 k^2 ZA+k^2 X+\frac{2}{3}k^2 Y-\frac{2}{3}k^2 k^2 YZ\right)u_k
\\
&+ik\left(k^2 Z-k^2 k^2 ZB +\frac{2}{3}k^2 YX\right)\rho_k\\
&+ik\left(k^2 Z-k^2 k^2 ZC+\frac{2}{3}k^2 Y^2 \right)T_k .\\
\end{split}
\end{equation*}

The invariance equation (\ref{generalIE}) for this case is a system of six coupled
quadratic equations with quadratic in $k^2$ coefficients:
\begin{equation}
\begin{split}\label{invariance131}
&F_1= -\frac{4}{3}-A-k^2(A^2+B-\frac{8Z}{15}+\frac{2C}{3})+\frac{2}{3}k^4CZ  =  0,\\
&F_2=                                   \frac{8}{15}X+B-A+k^2AB+\frac{2}{3}k^2CX  = 0, \\
& F_3=                                   \frac{8}{15}Y+C-A+k^2AC+\frac{2}{3}k^2CY  = 0, \\
& F_4=                       A+\frac{2}{3}Z+k^2ZA-X-\frac{2}{3}Y+\frac{2}{3}k^2YZ  = 0, \\
& F_5=                                 k^2B-\frac{2}{3}X-k^2Z+k^4ZB-\frac{2}{3}YX  = 0, \\
& F_6=                 -\frac{5}{2}-\frac{2}{3}Y+k^2(C-Z)+k^4ZC-\frac{2}{3}k^2Y^2  = 0.
\end{split}
\end{equation}
There are several approaches to to deal with this system. One can easily calculate the
Taylor series for $A,B,C,X,Y,Z$ in powers of $k^2$ at the point $k=0$. In application to
(\ref{form131}) this is exactly the Chapman--Enskog series (the Taylor series for
$\sigma$ and $q$). To find the linear and quadratic in $k$ terms in (\ref{form131}) we
need just a zeroth approximation for $A,B,C,X,Y,Z$ from (\ref{invariance131}):
\begin{equation}
A=B=-\frac{4}{3}, \; C=\frac{2}{3}, \; X=0, \; Y=-\frac{15}{4},\; Z=\frac{7}{4}.
\end{equation}
This is the Burnett approximation:
\begin{eqnarray*}\label{Burnett131}
   \sigma_{k} & = & -\frac{4}{3}iku_{k}+\frac{4}{3}k^2\rho_{k}-\frac{2}{3}k^2T_{k},\\
     q_{k} & = &-\frac{15}{4}ikT_{k}-\frac{7}{4}k^2 u_{k}
\end{eqnarray*}
The dispersion relation for this Burnett approximation coincides with the one obtained by Bobylev \cite{Bob} from
the Boltzmann equation for Maxwell molecules, and the short waves are unstable in this approximation.

Direct Newton's iterations produce more sensible results. Thus, starting from $A=B=C=X=Y=Z=0$ we get the
first iteration
\begin{eqnarray*}
\label{1N131}
A_1&=&-20\frac{141k^2 +20}{867k^2k^2 + 2105k^2 +300},\\\nonumber
B_1&=&-20\frac{459k^2 k^2 +810k^2 +100}{3468k^2 k^2 k^2
+12755k^2 k^2 +11725k^2 +1500},
\\\nonumber
C_1&=&-10\frac{51k^2 k^2 -485 k^2 -100}{3468k^2 k^2 k^2
+12755k^2 k^2 +11725k^2 +1500},
\\\nonumber
X_1&=&-\frac{375k^2 (21k^2 -5)}{2(3468k^2 k^2 k^2
+12755k^2 k^2 +11725k^2 +1500)},\\\nonumber
Y_1&=&-\frac{225(394k^2 k^2 +685 k^2 +100)}{4(
3468k^2 k^2 k^2 +12755k^2 k^2 +11725k^2 +1500)},\\\nonumber
Z_1&=&-15\frac{153k^2 +35}{867k^2k^2 + 2105k^2 +300}.
\end{eqnarray*}
The corresponding hydrodynamics are non-local but stable and were first obtained by a
partial summation (regularization) of the Chapman--Enskog series \cite{GKJETP91}.

A numerical solution of the invariance equation (\ref{invariance131}) is also
straightforward and does not produce any serious problem. Selection of the proper
(Chapman--Enskog) branch of the solution, is set by the asymptotics: $\omega\to 0$ when
$k\to 0$.

The dispersion equation for frequency $\omega$ is
\begin{equation}
\label{dispersion131}
\begin{split}
\omega^3&-k^2\left(\frac{2}{3} Y+ A\right) \omega^2 \\
&+k^2\left(\frac{5}{3}
-\frac{2}{3} k^2 Z-\frac{2}{3}k^2 C-k^2 B+\frac{2}{3}k^2 AY+\frac{2}{3} k^2 k^2 CZ
\right) \omega \\&+\frac{2}{3}k^2 (k^2 X-k^2 Y+k^2 k^2 BY -k^2 k^2 XC)=0.
\end{split}
\end{equation}

\begin{figure}[t]
\boxed{\includegraphics[width=0.8\textwidth]{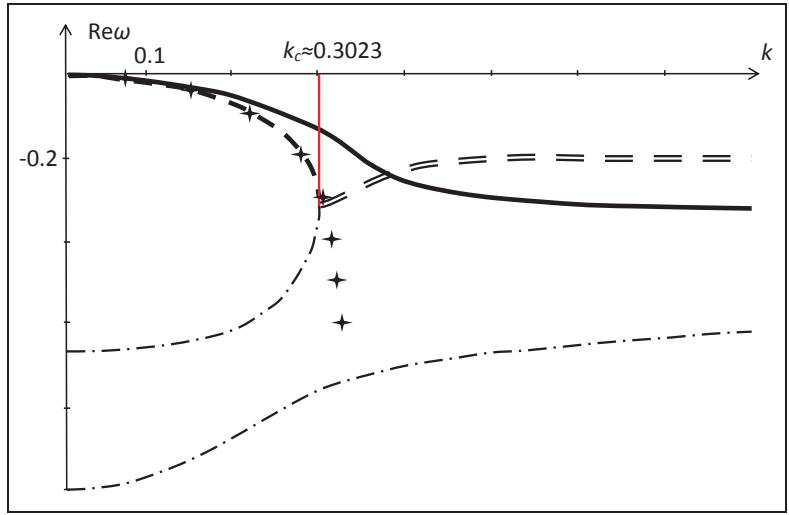}}
\caption{The dispersion relation for the linearized 1D
Grad system (\ref{Grad131}).
The solution for the whole kinetic system  (\ref{Grad131}) features five
$\omega$'s, while the motions on the hydrodynamic invariant manifold
has three of them for each $k<k_c$ and destroys
for $k\ge k_c$.  The bold solid line shows the hydrodynamic acoustic mode (two complex conjugated roots).
The bold dashed line for $k<k_c$ is the hydrodynamic diffusion mode (a real root).
At  $k=k_c$ this line meets a real root of nonhydrodynamic mode (thin dash-dot line) and
for  $k>k_c$ they turn into a couple of complex conjugated roots (bold double-dashed line at $k>k_c$).
The four-point stars correspond to the third Newton iteration for the diffusion mode.
A dash-and-dot line at the bottom of the plot shows the
isolated non-hydrodynamic mode (single real root of (\ref{dispersiongrad101}).\label{mkfigover}}
\end{figure}

The real-valued solution to the invariance equation (\ref{invariance131}) does not exist
for sufficiently large $k$. (A telling simple example of such a behavior of real
algebraic sets gives the equation $k^2(1-k^2)+A^2=0$.) Above a critical value $k_c\approx
0.3023$, the Chapman--Enskog branch in (\ref{invariance131}) disappears, and two complex
conjugated solutions emerge. This situation becomes clear if we look at the dispersion
curves (Fig.~\ref{mkfigover}). For $k<k_c$ the  Chapman--Enskog branch of the dispersion
relation consists of three hydrodynamic modes starting from 0 at $k=0$. Two
non-hydrodynamic modes start from strictly negative values at $k=0$ and are real-valued.
They describe the relaxation to the hydrodynamic invariant manifold from the initial
conditions outside this manifold. (This is, in other words, relaxation of the
non-hydrodynamic variables, $\sigma_k$ and $q_k$ to their values $\sigma_k (\rho_k ,u_k ,
T_k , k)$ and $q_k (\rho_k ,u_k , T_k , k)$.) For $k<k_c$ the non-hydrodynamic modes are
real-valued, the relaxation goes exponentially, without damped oscillations. At $k=k_c$,
one root from the non-hydrodynamic branch crosses a real-valued root of the hydrodynamic
branch and they together transform into a couple of complex conjugated roots when
$k>k_c$. It is impossible to capture two pairs of complex modes by an equation for three
macroscopic variables and, at the same time, it is impossible to separate two complex
conjugated modes between two systems of real-valued equations.

For small $k$, when the separation of time between the ``fast'' collision term and the
``not-so-fast'' advection is significant, there is an essential difference between the
relaxation of hydrodynamic and non-hydrodynamic variables: $\rho$ and $u$ do not change
in collision and their relaxation is relatively slow, but $\sigma$ and $q$ are directly
affected by collisions and their relaxation to  $\sigma_k (\rho_k ,u_k , T_k , k)$ and
$q_k (\rho_k ,u_k , T_k , k)$ is fast. Nevertheless, when $k$ grows and achieves $k_c$
the difference between the hydrodynamic and non-hydrodynamic variables becomes less
pronounced. In such a case, the 4-dimensional invariant manifold may describe the
relaxation better. For this purpose, we can create the invariance equation for an
extended list of four `hydrodynamic variables' and repeat the construction. Instead of
the selection of the Chapman--Enskog branch only, we have to select a continuous branch
which includes the roots with $\omega\to 0$ when $k\to 0$.

The 2D algebraic manifold given by the dispersion equation (\ref{dispersion131}) and the
invariance equation (\ref{invariance131})  represents the important properties of the
hydrodynamic invariant manifold (see Fig~\ref{mkfigover}). In particular, the crucial
question is the existence of the Chapman--Enskog branch and the description of the
connected component of this curve which includes the germ of the Chapman--Enskog branch
near $k=0$.

Iterations of the Newton method for the invariance equation converge fast to the solution
with singularity. For $k<k_c$ the corresponding attenuation rates converge to the exact
solution and for  $k> k_c$ the real part of the diffusion mode ${\rm Re} \omega \to
-\infty$ with Newton's iterations (Fig.~\ref{mkfigover}). The corresponding limit system
has the infinitely fast decay of the diffusion mode when $k> k_c$. This regularization of
singularities by the infinite dissipation is quite typical for the application of the
Newton method to solution of the invariance equation. The `solid jet' limit for the
extremely fast compressions gives us another example \cite{GKsolidliq1995} (see also
Sec.~\ref{IMBoltzmann}).

\subsection{Invariant manifolds, entanglement of hydrodynamic and non-hydro\-dynamic modes
and saturation of dissipation for the 3D 13 moments Grad system
\label{Sec:Grad3DDestroy}}

The thirteen moments linear Grad system consists of 13 linearized PDE's
(\ref{balanceequations}), (\ref{Grad133e}) giving the time evolution of the hydrodynamic
fields (density $\rho$, velocity vector field $\uu$, and temperature $T$) and of
higher-order distinguished moments: five components of the symmetric traceless stress
tensor $\s$ and three components of the heat flux $\qq$ \cite{Grad}. With this example,
we conclude the presentation of exact hydrodynamic manifolds for linearized Grad models.

A point of departure is the Fourier transform of the linearized three-dimensional Grad's
thirteen-moment system:
\begin{eqnarray*}
 \partial_t \rho_k &=& -i\kk\cdot\uu_k, \\
 \partial_t \uu_k &=& -i\kk \rho_k - i \kk T_k - i \kk\cdot\s_k,  \\
 \partial_t T_k &=& -\frac{2}{3}i \kk\cdot(\uu_k + \qq_k), \\
 \partial_t \s_k &=& -2 i \overline{\kk\uu_k} -
 \frac{4}{5} i\overline{\kk \qq_k} -
 \s_k , \\
 \partial_t \qq_k &=& -\frac{5}{2} i \kk T_k - i \kk\cdot\s_k -
 \frac{2}{3}\qq,
\end{eqnarray*}
where $\kk$ is the wave vector, $\rho_k$, $\uu_k$ and $T_k$ are the Fourier images
for density, velocity and temperature, respectively, and $\s_k$ and $\qq_k$ are the nonequilibrium traceless
symmetric stress tensor ($\overline{\s}=\s$) and heat flux vector components,
respectively.

Decompose the vectors and tensors into the parallel (longitudinal) and orthogonal (lateral)
parts with respect to the wave vector $\kk$, because the fields are rotationally
symmetric around any chosen direction $\kk$. A unit vector in the direction of the wave
vector is $\e=\kk/k$, $k=|\kk|$, and the corresponding decomposition is $\uu_k = u_k^\|
\,\e + \uu_k^\perp$, $\qq_k = q_k^\| \,\e + \qq_k^\perp$, and $\s_k = \frac{3}{2}
\sigma_k^\| \overline{\e\e} + 2\s_k^\perp$, where $\e\cdot\uu_k^\perp=0$,
$\e\cdot\qq_k^\perp=0$, and $\e\e:\s_k^\perp=0$.

In these variables, the linearized 3D 13-moment Grad system decomposes into two closed
sets of equations, one  for the longitudinal and another for the lateral modes. The
equations for $\rho_k$, $u_k^\|$, $T_k$, $\sigma_k^\|$, and $q_k^\|$ coincide with the 1D
Grad system (\ref{Grad131}) from the previous section (the difference is just in the
superscript $^\|$). For the lateral modes we get
\begin{equation}
\begin{split}
 \partial_t \uu_k^\perp &= - i k\, \e\cdot\s_k^\perp, \\
 \partial_t \s_k^\perp &= -i k \overline{\e\uu_k^\perp} - \frac{2}{5} i k \overline{\e\qq_k^\perp} -
  \s_k^\perp, \\
 \partial_t \qq_k^\perp &= - i k \,\e\cdot\s_k^\perp - \frac{2}{3} \qq_k^\perp.
 \label{reducedsetB}
\end{split}
\end{equation}
The hydrodynamic invariant manifold for these decoupled systems is a direct product of
the invariant manifolds for (\ref{Grad131}) and for (\ref{reducedsetB}). The
parametrization (\ref{form131}), the invariance equation (\ref{invariance131}), the
dispersion equation for exact hydrodynamics (\ref{dispersion131}) and the plots of the
attenuation rates (Fig.~\ref{mkfigover}) for (\ref{Grad131}) are presented in the
previous section.

For the lateral modes the hydrodynamic variables consist of the 2D vector $\uu_k^\perp$.
We use the general expression (\ref{liftingGen}) and take into account the rotational
symmetry for the parametrization of the non-hydrodynamic variables $\s_k^\perp$ and
$\qq_k^{\bot}$ by the hydrodynamic ones:
\begin{equation} \label{introDU}
 \s_k^\perp = i k D(k^2) \overline{\e\uu_k^\perp}, \;\;
 \qq_k^{\bot} = -k^2 U(k^2) \uu_k^\perp .
\end{equation}

\begin{figure}[t]
\boxed{\includegraphics[width=0.8\textwidth]{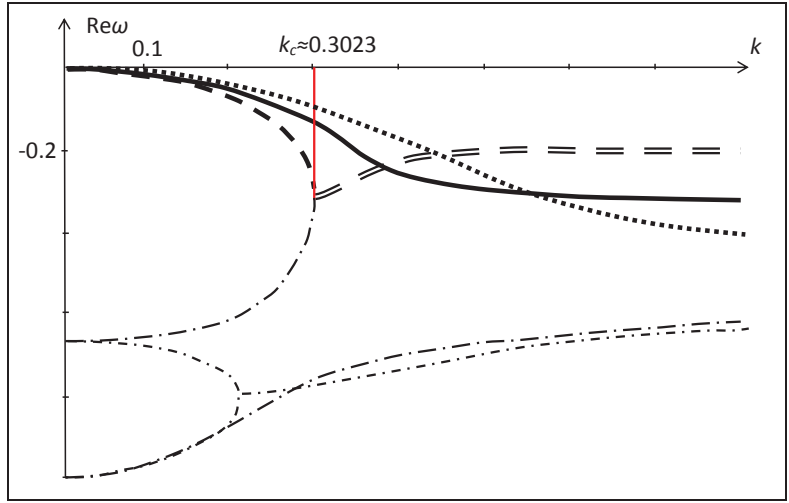}}
\caption{The dispersion relation for the linearized 3D 13 moment
Grad system (\ref{balanceequations}), (\ref{Grad133e}).
 The bold solid line shows the hydrodynamic acoustic mode (two complex conjugated roots). The bold dotted line
represents the shear mode (double degenerated real-valued root).
The bold dashed line for $k<k_c$ is the hydrodynamic diffusion mode (a real-valued root).
At  $k=k_c$ this line meets a real-valued root of non-hydrodynamic mode (thin dash-and-dot line) and
for  $k>k_c$ they turn into a couple of complex conjugated roots (bold double-dashed line at $k>k_c$).
Dash-and-dot lines at the bottom of the plot show the
separated non-hydrodynamic modes. All the modes demonstrate the saturation of dissipation.\label{FigGrad3D13M}}
\end{figure}

There are two unknown scalar real-valued functions here:  $D(k^2)$ and $U(k^2)$. We
equate the microscopic and macroscopic time derivatives of the non-hydrodynamic
variables and get the invariance conditions:
\begin{equation}
\begin{split}
&\frac{\partial  \s_k^ \perp}{\partial
\uu_k^\perp}\cdot(-ik\e\cdot \s_k^\perp)=-i k \overline{\e\uu_k^\perp} - \frac{2}{5} i k \overline{\e\qq_k^\perp} -
  \s_k^\perp,\\
&\frac{\partial  \qq_k^\perp}{\partial \uu_k^\perp}\cdot(-ik\e\cdot
\s_k^\perp)=- i k \,\e\cdot\s_k^\perp - \frac{2}{3} \qq_k^\perp, \label{invDU}
\end{split}
\end{equation}
We substitute here $\s_k^\perp$ and  $\qq_k^\perp$ by the expressions (\ref{introDU}) and
derive the algebraic invariance equation for  $D$ and $U$, which can be transformed into
the form:
\begin{equation}
\begin{split}
& 15 k^4 D^3 + 25 k^2 D^2 + (10+21 k^2) D + 10 = 0, \\
&U=-\frac{3D}{2+3k^2D}. \label{DU}
\end{split}
\end{equation}
The solution of the cubic equation (\ref{DU}) with the additional condition $D(0)=-1$
matches the Navier-Stokes asymptotics and  is real-valued for all $k^2$
\cite{ColKarKroe2007_2}. The dispersion equation gives a twice-degenerated real-valued
shear mode. All 13 modes for the three-dimensional, 13 moment linearized Grad system are
presented in Fig.~\ref{FigGrad3D13M} with 5 hydrodynamic and 8 non-hydrodynamic modes.
This plot includes also 8 modes (3 hydrodynamic and 5 non-hydrodynamic ones) for the
one-dimensional system (\ref{Grad131}). Entanglement between hydrodynamic and
non-hydrodynamics modes appears at the same critical value of $k\approx 0.3023$ and the
exact hydrodynamics does not exist for larger $k$.

\subsection{Algebraic hydrodynamic invariant manifold for the linearized Boltzmann and BGK
equations: separation of hydrodynamic and non-hydro\-dynamic modes\label{Sec:BGKexist}}

The entanglement of the hydrodynamic and non-hydrodynamic modes  at large wave vectors
$k$ destroys the exact hydrodynamic for the Grad moment equations. We conjecture that
this is the catastrophe of the applicability of the moment equations and the hydrodynamic
manifolds are destroyed together with the Grad approximation. It is plausible that if the
linearized collision operator has a spectral gap (see a review in \cite{Mouhot2006})
between the five time degenerated zero and other eigenvalues then the algebraic
hydrodynamic invariant manifold exists for all $k$. This remains an open question but the
numerical calculations of the hydrodynamic invariant manifold available  for the
linearized kinetic equation (\ref{BOL0}) with the  BGK collision operator
\cite{BGK,GKMod} support this conjecture \cite{KarColKroe2008PRL}.

The incompressible hydrodynamic limit for the scaled solutions of the BGK equation was proven in 2003
\cite{Saint-RaymondBGK2003}.

The linearized kinetic equation (\ref{BOL0}) has the form
\begin{equation}
    \partial_t f + \vv\cdot \nabla_x f= Lf,
\end{equation}
where $f(t,\vv,\xx)$ is the {\em deviation} of the distribution function from its
equilibrium value $f^*(\vv)$, $L$ is the linearized kinetic operator. Operator $L$ is
symmetric with respect to the entropic inner product
\begin{equation}\label{LinKin}
\langle\varphi, \psi\rangle_{f^*}= \int \frac{\varphi (\vv) \psi (\vv)}{f^*(\vv)}\, \D^3 \vv.
\end{equation}
In the $L_2$ space with this inner product, $\ker L=({\rm im} L)^\bot$ is a finite
dimensional subspace. It is spanned by five functions
$$f^*(\vv), \vv f^*(\vv), v^2 f^*(\vv).$$
The hydrodynamic variables (for the given $t$ and $x$) are the inner products of these
functions on $f(t,\xx,\vv)$, but it is more convenient to use the orthonormal basis with
respect to the product $\langle \cdot , \cdot \rangle_{f*}$, $\varphi_1(\vv), \ldots,
\varphi_5(\vv)$. The macroscopic variables are $M_i=\langle\varphi_i, f\rangle_{f^*}$
($i=1,2,\ldots, 5$).

It is convenient to represent $f$ in the form of the direct sum of the macroscopic
and microscopic components
$$f=P_{\rm macro} f + P_{\rm micro} f,$$
where $$P_{\rm macro} f =\sum_i \varphi_i \langle \varphi_i, f \rangle_{f^*}, \;\; P_{\rm
micro} f= (f-\sum_i \varphi_i \langle \varphi_i, f \rangle_{f^*})$$

After the Fourier transformation the linearized kinetic equation is
\begin{equation}
\partial_t f_k =-i (\kk,\vv) f_k +Lf_k .
\end{equation}
The lifting operation $\mathcal{X}(\kk):M_k \mapsto f_k$ (\ref{generalAnz}) should have
the form
$$\mathcal{X}(\kk)(M)=\sum_i M_{ik} \varphi_i (\vv) + \sum_i M_{ik}\psi_i(\kk,\vv),$$
where  $\langle \varphi_i, \psi_j \rangle_{f^*}=0$ for all $i,j=1,2,\ldots, 5$. We equate
the microscopic and macroscopic time derivatives (\ref{micrimacrogen}) of $f$ and get the
invariance equation (\ref{generalIE}):
\begin{equation}\label{BGKinvariance}
\begin{split}
L\psi_j=&i\kk\cdot [P_{\rm micro} (\vv\varphi_j)+P_{\rm micro} (\vv \psi_j)\\
&- \sum_l \psi_l\langle\varphi_l,\vv\varphi_j\rangle_{f^*}
-\sum_l \psi_l \langle \varphi_l,\vv \psi_j \rangle_{f^*}].
\end{split}
\end{equation}
For the solution of this equation, it is important that  ${\rm im} L={\rm im} P_{\rm micro}$ and the
both operators $L$ and $L^{-1}$ are defined and bounded on this microscopic subspace. The
linearized BGK collision integral is simply $L=- P_{\rm micro}$ (the relaxation parameter
$\epsilon=1$) and the invariance equation has in this case an especially simple form.

In \cite{KarColKroe2008PRL} the form of this equation has been analyzed further and it
has been solved numerically by several methods: the Newton iterations and continuation in
parameter $k$. The attenuation rates for the Chapman--Enskog branch have been analyzed.
All the methods have produced the same results: (i) the real-valued hydrodynamic invariant
manifold exists  for all range of  $k$, from zero to large values, (ii) hydrodynamic
modes are always separated from the  non-hydrodynamic modes (no entanglement effects),
and (iii) the saturation of dissipation exists for large $k$.

\section{Hydrodynamic invariant manifolds for nonlinear kinetics\label{Sec:NlinIM}}

\subsection{1D nonlinear Grad equation and nonlinear viscosity\label{Sec:NlinVis}}

In the preceding sections we represented the hydrodynamic invariant manifolds for {\em
linear } kinetic equations. The algebraic equations for these manifolds in $k$-space have
a relatively simple closed form and can be studied both analytically and numerically. For
nonlinear kinetics, the situation is more difficult for a simple reason: it is impossible
to cast the problem of the invariant manifold in the form of a system of decoupled
finite--dimensional problems by the Fourier transform. The equations for the invariant
manifolds for the finite--dimensional nonlinear dynamics have been published by Lyapunov
in 1892 \cite{Lya1992} but even for ODEs this is a nonlinear and rather non-standard
system of PDEs.

There are several ways to study the hydrodynamic invariant manifolds for the nonlinear
kinetics. In addition to the classical Chapman--Enskog series expansion, we can solve the
invariance equation numerically or semi--analytically, for example, by the iterations
instead of the power series. In the next section, we demonstrate this method for the
Boltzmann equation. In this section, we follow the strategy that seems to be promising:
to evaluate the asymptotics of the hydrodynamic invariant manifolds at large gradients
and frequencies and to match these asymptotics with the first Chapman--Enskog terms. For
this purpose, we use exact summation of the ``leading terms" in the Chapman--Enskog
series.

The starting point is the set of the 1D nonlinear Grad equations for the
hydrodynamic variables $\rho $, $u$ and $T$, coupled with the
non-hydrodynamic variable $\sigma$, where $\sigma$ is the $xx$-component of the stress
tensor:
\begin{eqnarray}
\label{Gradnl}
\partial_t \rho &=&-\partial_x (\rho u); \label{Ga}\\
\partial_t u &=&-u\partial_x u -\rho^{-1}\partial_x p
-\rho^{-1}\partial_x \sigma; \label{Gb}\\
\partial_t T &=&-u\partial_x T-(2/3)T\partial_x u
-(2/3)\rho^{-1}\sigma\partial_x u; \label{Gc}\\
\partial_t \sigma &=&-u\partial_x \sigma -(4/3)p\partial_x u
-(7/3)\sigma\partial_x u -\frac{p}{\mu(T)}\sigma .\label{Gd}
\end{eqnarray}
Here $p=\rho T$ and $\mu(T)$ is the temperature-dependent viscosity coefficient. We adopt
the form $\mu(T)=\alpha T^{\gamma}$, where
$\gamma$ varies from $\gamma=1$ (Maxwell's molecules) to $\gamma=1/2$ (hard spheres) \cite{Chapman}.

Our goal is to compute the correction to the Navier--Stokes approximation of the
hydrodynamic invariant manifold, $\sigma_{\rm NS}=-(4/3)\mu\partial_x u$, for high values
of the velocity. Let us consider first the Burnett correction from
(\ref{Gradnl})-(\ref{Gd}):
\begin{equation}
\label{Burnett}
\sigma_{\rm B} =-\frac{4}{3}\mu\partial_x u +\frac{8(2-\gamma)}{9}\mu^2
p^{-1}(\partial_x u)^2-\frac{4}{3}\mu^2 p^{-1}\partial_x
(\rho^{-1}\partial_x p).
\end{equation}
Each further $n$th term of the Chapman--Enskog expansion contributes, among others, a
nonlinear term proportional to $(\partial_x u)^{n+1}$. Such terms can be named the {\it
high-speed} terms since they dominate the rest of the contributions in each  order of the
Chapman--Enskog expansion when the characteristic average velocity is comparable to the
thermal speed. Indeed, let $U$ be the characteristic velocity (the Mach number). Consider
the scaling  $u=U \widetilde{u} $, where $\widetilde{u}=O(1)$. This velocity scaling is
instrumental to the selection of the leading large gradient terms and the result below is
manifestly Galilean--invariant.

The term $(\partial_x u)^{n+1}$ includes the factor $U^{n+1}$ which is the highest
possible order of $U$ among the terms available in the $n$th order of the Chapman--Enskog
expansion. Simple dimensional analysis leads to the conclusion that such terms are of the
form $$\mu (p^{-1}\mu\partial_x u)^{n}\partial_x u =\mu g^{n}\partial_x u, $$ where
$g=p^{-1}\mu\partial_x u$ is dimensionless. Therefore, the Chapman--Enskog expansion for
the function $\sigma$ may be formally rewritten as:
\begin{equation}
\label{representation}
\sigma\!=\!-\!\mu\!\left\{\!\frac{4}{3}\!-\!\frac{8(2\!-\!\gamma)}{9}g\!
+\!
r_2g^2\!+\!
\dots\!+\!r_ng^n\!+\!\dots\!\right\}
\!\partial_xu\!+\!\dots\!
\end{equation}
The series in the brackets is the collection of the high-speed contributions of interest,
coming from {\it all} orders of the Chapman--Enskog expansion, while the dots outside the
brackets stand for the terms of other natures. Thus after summation the series of the
high-speed corrections to the Navier--Stokes approximation for the Grad equations
(\ref{Gradnl}) takes the form:
\begin{equation}
\label{R}\boxed{
\sigma_{\rm nl}=-\mu R(g)\partial_x u ,}
\end{equation}
where $R(g)$ is a yet unknown function represented by a formal subsequence of
Chapman--Enskog terms in the expansion (\ref{representation}). The function $R$ can be
considered as a dynamic modification of the viscosity $\mu$ due to the gradient of the
average velocity.

Let us write the invariance equation for the representation (\ref{R}). We first compute
the microscopic derivative of the function $\sigma_{\rm nl}$ by substituting (\ref{R})
into the right hand side of (\ref{Gd}):
\begin{equation}
\begin{split}
\label{microNL}
\partial_t^{\rm micro}\sigma_{\rm nl}&=
-u\partial_x \sigma_{\rm nl}-\frac{4}{3}p\partial_x u -\frac{7}{3}\sigma_{\rm nl}\partial_x u
-\frac{p}{\mu(T)}\sigma_{\rm nl} \\
&=\left\{-\frac{4}{3}+\frac{7}{3}gR+R\right\}p\partial_x u +\dots,
\end{split}
\end{equation}
where dots denote the terms irrelevant to the high speed approximation (\ref{R}).

Second, computing the macroscopic derivative of $\sigma_{\rm nl}$ due to (\ref{Ga}),
(\ref{Gb}), and (\ref{Gc}), we obtain:
\begin{equation}
\label{macroNL1}
\partial_t^{\rm macro}\sigma_{\rm nl}=-[ \partial_t \mu(T)]R\partial_x
u-\mu(T)\frac{\D R}{\D g}[\partial_t g]\partial_x u- \mu(T)R\partial_x [\partial_t u].
\end{equation}
In the latter expression, the time derivatives of the hydrodynamic variables should be
replaced with the right hand sides of (\ref{Ga}), (\ref{Gb}), and (\ref{Gc}), where, in
turn,  $\sigma$ should be replaced by  $\sigma_{\rm nl}$ (\ref{R}). We find:
\begin{equation}
\label{macroNL2}
\partial_t^{\rm macro}\sigma_{\rm nl}=\left\{gR+\frac{2}{3}(1-gR)\times
\left(\gamma gR+(\gamma-1)g^2 \frac{\D R}{\D g}\right)\right\}p\partial_x u +\dots
\end{equation}
Again we omit the terms irrelevant to the analysis of the leading terms.

Equating the relevant terms in (\ref{microNL}) and (\ref{macroNL2}), we obtain the
approximate invariance equation for the function $R$:
\begin{equation}\boxed{
\label{AnnPhysinvariance} (1-\gamma)g^2 \left(1-gR\right)\frac{\D R}{\D g}+\gamma g^2 R^2
+\left[\frac{3}{2}+g(2-\gamma)\right]R-2=0.}
\end{equation}
It is approximate because in the microscopic derivative many terms are omitted, and it
becomes more accurate when the velocities are multiplied by a large factor. When $g \to
\pm \infty$ then the viscosity factor (\ref{AnnPhysinvariance}) $R \to 0$.

For Maxwell's molecules ($\gamma=1$), (\ref{AnnPhysinvariance}) simplifies
considerably, and becomes the algebraic equation:
\begin{equation}
\label{Maxwell}
g^2 R^2 +\left(\frac{3}{2}+g\right)R-2=0.
\end{equation}
The solution recovers the Navier--Stokes relation in the limit of small $g$ and for an
arbitrary $g$ it reads:
\begin{equation}
\label{viscosity}
R_{\rm MM}=\frac{-3-2g+3\sqrt{1+(4/3)g+4g^2 }}{4g^2 }.
\end{equation}
The function $R_{\rm MM}$ (\ref{viscosity}) is plotted in Fig.~\ref{AnnPhysFig10}. Note
that $R_{\rm MM}$ is positive for all values of its argument $g$, as is appropriate for
the viscosity factor, while the Burnett approximation to the function $R_{\rm MM}$
violates positivity.

\begin{figure}[t]
\centering{\boxed{
\includegraphics[width=0.7 \textwidth]{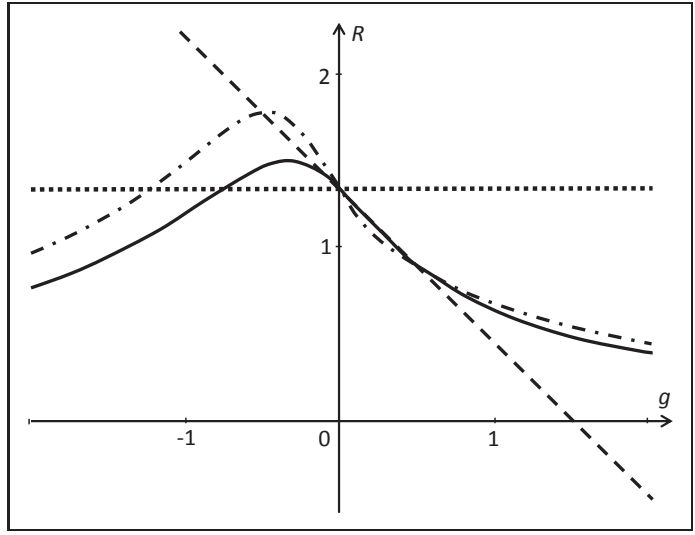}}
\caption{\label{AnnPhysFig10}Viscosity factor $R(g)$ (\ref{AnnPhysinvariance}):  solid - $R(g)$ for Maxwell
molecules; dash - the Burnett approximation of $R(g)$ for Maxwell
molecules; dots - the Navier--Stokes approximation; dash-dots - Viscosity factor $R(g)$ for hard spheres,
the first approximation (\ref{approximation}).}}
\end{figure}

For other models ($\gamma\ne 1$), the invariance equation (\ref{AnnPhysinvariance}) is a
 nonlinear ODE with the initial condition $R(0)=4/3$ (the
Navier--Stokes condition). Several ways to derive analytic results are possible. One
possibility is to expand the function $R$ into powers of $g$, around the point $g=0$.
This brings us back to the original sub-series of the Chapman--Enskog expansion
(\ref{representation})). Instead, we take advantage of the opportunity offered by the
parameter $\gamma$. Introduce another parameter $\beta=1-\gamma$, and consider the
expansion:
\[ R(\beta,g)=R_0 (g)+\beta R_1 (g)+\beta^2 R_2 (g)+ \dots .\]
Substituting this expansion into the invariance equation (\ref{AnnPhysinvariance}), we
derive $R_0 (g)=R_{\rm MM}(g)$,
\begin{equation}
\label{approximation} R_1 (g)=-g(1-gR_0 )\frac{R_0 +g(\D R_0 /\D g)}{2g^2 R_0 +g +(3/2)},
\end{equation}
etc. That is, the  solution for models different from Maxwell's molecules is constructed
in the  form of a series with the exact solution for the Maxwell molecules as the leading
term. For hard spheres ($\beta=1/2$), the result to the first-order term reads: $R_{\rm
HS}\approx R_{\rm MM}+(1/2)R_1$. The resulting approximate viscosity factor is shown in
Fig.~\ref{AnnPhysFig10} (dash-dots line). The features of the approximation obtained are
qualitatively the same as in the case of Maxwell molecules.

Precisely the same result for the nonlinear elongational viscosity obtained first from
the Grad equations \cite{KDNPRE97} was derived in \cite{Santos1,Santos3} from the
solution to the BGK kinetic equation in the regime of so-called homoenergetic extension
flow. This remarkable fact gives more credit to the derivation of hydrodynamic manifolds
from nonlinear Grad equations.

The approximate invariance equation (\ref{AnnPhysinvariance}) defines the relevant
physical solution to the viscosity factor for all values of $g$. The hydrodynamic
equations are now given by (\ref{Ga}), (\ref{Gb}), and (\ref{Gc}), where $\sigma$ is
replaced by $\sigma_{\rm nl}$ (\ref{R}). First, the correction concerns the nonlinear
regime, and, thus, the linearized form of the new equations coincides with the linearized
Navier--Stokes equations. Second, the solution (\ref{viscosity}) for Maxwell molecules
and the result of the approximation (\ref{approximation}) for other models and also the
numerical solution \cite{KGAnPh2002} suggest that the modified viscosity $\mu R$ vanishes
in the limit of very high values of the velocity gradients.  However, a cautious remark
is in order since the original ``kinetic'' description is Grad's equations
(\ref{Gradnl})-(\ref{Gd}) and not the Boltzmann equation. The first Newton iteration for
the Boltzmann equation gives a singularity of viscosity at a large negative value of
divergency (see below, Sec.~\ref{IMBoltzmann}).

\subsection{Approximate invariant manifold for the Boltzmann equation\label{IMBoltzmann}}

\subsubsection{Invariance equation}

We  begin with writing down the invariance condition for the hydrodynamic manifold of the Boltzmann equation.
A convenient point of departure is the Boltzmann equation (\ref{BOL0}) in a co-moving reference frame,

\begin{equation}\label{BPLTZ}
    D_t f=- (\vv-\uu) \cdot \nabla_x f+Q(f),
\end{equation}
where $D_t$ is the material time derivative,
$D_t=\partial_t+\uu\cdot\nabla_x .$
The macroscopic (hydrodynamic) variables are:
$$M= \left\{n;n\uu; \frac{3nk_{\rm B}T}{\mu}+nu^2\right\}=m[f]=\int \{1; \vv; v^2 \} f \, \D \vv,$$
where $n$ is number density, $\uu$ is the flow velocity, and $T$ is the temperature; $\mu$ is particle's mass and $k_{\rm B}$ is Boltzmann's constant.
These fields do not change in collisions, hence, the projection of the Boltzmann equation
on the hydrodynamic variables is
\begin{equation}\label{BoltzMomentProj}
D_t M = -  m[(\vv-\uu) \cdot\nabla_x f].
\end{equation}

For the given hydrodynamic fields $M$ the local Maxwellian $f^{\rm LM}_M$ (or just
$f^{\rm LM}$) is the only zero of the collision integral $Q(f)$.
\begin{equation}
f^{\rm LM}=n\left(\frac{2\pi k_{\rm B}T}{\mu}\right)^{-3/2} \exp \left(-\frac{\mu(\vv-\uu)^{2}}{2k_{\rm B}T}\right). \label{4.1}
\end{equation}
The local Maxwellian depends on space through the hydrodynamic fields.

We are looking for an invariant manifold $\ff_M$ in the space of distribution functions parameterized by the
hydrodynamic fields.
Such a manifold is represented by a lifting map $M \mapsto f_M$
that maps the hydrodynamic fields in 3D space, three functions of the space variables,
$M=\{n(\xx), \uu (\xx), T(\xx)\}$, into a function of
six variables $\ff_M(\xx, \vv)$. The consistency condition should hold:
\begin{equation}
m[f_{M}]=M.
\end{equation}
The differential of the lifting operator at the point $M$ is a linear map $(D_M \ff_M):\delta
M \to \delta f$.

It is straightforward to write down the invariance condition for the hydrodynamic manifold:
The microscopic time derivative of $f_M$ is given by the right hand side of the Boltzmann equation on the manifold,
$$D^{\rm micro} _t f_M = - (\vv-\uu) \cdot \nabla_x f_M + Q(f_M),$$
while the macroscopic  time derivative is defined by the chain rule:
$$D^{\rm macro} _t f_M = - (D_M f_M)m[(\vv-\uu) \cdot \nabla_x f_M].$$
The invariance equation  requires that, for any $M$, the outcome of two ways of taking the derivative should be the same:
\begin{equation}\label{invBE}\boxed{
-(D_M \ff_M)m[(\vv-\uu) \cdot \nabla_x \ff_M]=- (\vv-\uu) \cdot \nabla_x \ff_M + Q(\ff_M)}.
\end{equation}

One more field plays a central role in the study of invariant manifolds, the {\em defect of
invariance}:
\begin{equation}\label{Defect}
\begin{split}
\Delta_M&=D^{\rm macro} _t \ff_M-D^{\rm micro} _t \ff_M \\
&= -(D_M \ff_M)m[(\vv-\uu)
\cdot \nabla_x \ff_M] + (\vv-\uu) \cdot \nabla_x \ff_M.
\end{split}
\end{equation}
It measures the  ``non-invariance'' of a manifold $\ff_M$.

Let an approximation of the lifting operation $M\to \ff_M$ be given. The equation of the
first iteration for the unknown correction $\delta \ff_M$ of $\ff_M$ is obtained by the
linearization (We assume that the initial approximation, $\ff_M$, satisfies the
consistency condition and $m[\delta f]=0$.):
\begin{equation}\label{LinIEBoltzmann}
(D_M \ff_M)m[(\vv-\uu) \cdot \nabla_x \delta \ff_M]-
(\vv-\uu) \cdot \nabla_x \delta \ff_M + L\delta \ff_M=\Delta_M.
\end{equation}
Here, $L_M$ is a linearization of $Q$ at $f_M^{\rm LM}$. If $\ff_M$ is a local equilibrium
then, the integral operator $L_M$ at each point $\xx$ is symmetric with respect to the
entropic inner product (\ref{LinKin}). The equation of  iteration (\ref{LinIEBoltzmann})
is linear but with non-constant in space coefficients because both $(D_M \ff_M)$ and
$L_M$ depend on $\xx$.

It is necessary to stress that the standard Newton method does not work in these
settings. If $\ff_M$ is not a local equilibria then $L_M$ may be not  symmetric and we may
lose such instruments as the Fredholm alternative. Therefore, we use in the iterations
the linearized operators $L_M$ at the local equilibrium and not at the current
approximate distribution $\ff_M$ (the Newton--Kantorovich method).  We also do not
include the differential of the term $(D_M \ff_M)m$ in (\ref{LinIEBoltzmann}). The reason
for this {\em incomplete linearization} of the invariance equation (\ref{invBE}) is that
it provides convergence to the slowest invariant manifold (at least, for linear vector
fields), and other invariant manifolds are unstable in iteration dynamics. The complete
linearization does not have this property \cite{GKTTSP94,GorKarLNP2005}.

\subsubsection{Invariance correction to the local Maxwellian}

Let us choose the local Max\-wel\-lian $\ff_M=f^{\rm LM}$ (\ref{4.1}) as the initial approximation to the invariant
manifold in (\ref{LinIEBoltzmann}). In order to find the right hand side of this equation, we evaluate the
defect of invariance (\ref{Defect}) $\Delta_M=\Delta^{\rm LM}$:
\begin{equation}\label{LMdefect}
\Delta^{\rm LM}=f^{\rm LM}D,
\end{equation}
where
\begin{equation}\label{DLMdefect}\boxed{
\begin{aligned}
D=&\left(\frac{\mu(\vv-\uu)^2}{2k_{\rm B}T}-\frac{5}{2}\right)(\vv-\uu)\cdot\frac{\nabla_x T}{T}\\
&+\frac{\mu}{k_{\rm B}T}\left[(\vv-\uu)\otimes(\vv-\uu)-\frac{1}{3}\II(\vv-\uu)^2\right]:\nabla_x \uu.
\end{aligned}}
\end{equation}
Note that there is no ``smallness'' parameter
involved in the present consideration, the defect of invariance of the local Maxwellian is neither ``small'' or ``large''
by itself. We now proceed with finding a correction $\delta f$ to the local
Maxwellian on the basis of the linearized equation (\ref{LinIEBoltzmann})
supplemented with the consistency condition,
\begin{equation}\label{consistency1}
m[\delta f]=0.
\end{equation}
Note that, if we introduce the formal large parameter, $L\leftarrow \epsilon^{-1}L$ and look at the
leading-order correction $\delta f\leftarrow  \epsilon \delta f$,
disregarding all the rest in equation (\ref{LinIEBoltzmann}), we get a linear non-homogeneous integral equation,
\begin{equation}\label{CE1Boltzmann}
\begin{split}
\Lambda (\delta f/f^{\rm LM})=&\left(\frac{\mu(\vv-\uu)^2}{2k_{\rm B}T}-\frac{5}{2}\right)(\vv-\uu)\cdot\frac{\nabla_x T}{T} \\
&+\frac{\mu}{k_{\rm B}T}\left[(\vv-\uu)\otimes(\vv-\uu)-\frac{1}{3}\II(\vv-\uu)^2\right]:\nabla_x \uu,
\end{split}
\end{equation}
where
\begin{equation*}
\Lambda\varphi=\int w(\vv',\vv_1'|\vv,\vv_1)f^{\rm LM}(\vv_1)
[\varphi(\vv_1')+\varphi(\vv')-\varphi(\vv)-\varphi(\vv_1)]\D \vv_1'\D \vv'\D \vv_1
\end{equation*}
is the linearized Boltzmann collision operator ($w$ is the scattering kernel; standard
notation for the velocities before and after the binary encounter is used). It is readily
seen that (\ref{CE1Boltzmann}) is nothing but the standard equation of the first
Chapman-Enskog approximation, whereas the consistency condition (\ref{consistency1})
results in the unique solution (Fredholm alternative) to (\ref{CE1Boltzmann}). This leads
to the classical Navier-Stokes-Fourier equations of the Chapman-Enskog method.

Thus, the first iteration (\ref{LinIEBoltzmann}) for the solution of the invariance equation (\ref{invBE}) with the local Maxwellian as the initial approximation is matched to the first Chapman-Enskog correction to the local Maxwellian. However, equation (\ref{LinIEBoltzmann}) is much more complicated than its Chapman-Enskog limit: equation (\ref{LinIEBoltzmann}) is linear but integro-differential (rather than just the linear integral equation (\ref{CE1Boltzmann})), with coefficients varying in space through
both $(D_M f^{\rm LM})$ and $L$. We shall now describe a micro-local approach for solving (\ref{LinIEBoltzmann}).

\subsubsection{Micro-local techniques for the invariance equation}

Introducing $\delta f=f^{\rm LM}\varphi$, equation (\ref{LinIEBoltzmann}) for the local
Maxwellian initial approximation can be cast in the following form,
\begin{equation}\label{myform}
\Lambda^*\varphi-(\VV^*\cdot\nabla)\varphi=D,
\end{equation}
where the {\em enhanced linearized collision integral} $\Lambda^*$ and the {\em enhanced free flight operator} $(\VV^*\cdot\nabla)$ act as follows:
Let us denote $\Pi$ the projection operator ($\Pi^2=\Pi$),
\begin{equation}\label{LMproj}
\Pi g=\left(f^{\rm LM}\right)^{-1}D_Mf^{\rm LM}m[f^{\rm LM}g].
\end{equation}
Then in (\ref{myform}) we have:
\begin{eqnarray*}
\Lambda^*\varphi&=&\Lambda \varphi +(\Pi-1)(r\varphi),\label{enhancedLinColl}\\
r&=&(\vv-\uu)\cdot\frac{\nabla_x n}{n}+\frac{\mu}{k_{\rm B}T}(\vv-\uu)\otimes(\vv-\uu):\nabla_x \uu \\ &&+
\left(\frac{\mu(\vv-\uu)^2}{2k_{\rm B}T}-\frac{3}{2}\right)(\vv-\uu)\cdot\frac{\nabla_x T}{T}, \nonumber \\
(\VV^*\cdot\nabla)\varphi &=&(1-\Pi)((\vv-\uu)\cdot\nabla_x\varphi).\label{enhancedFreeFlight}
\end{eqnarray*}

The structure of the invariance equation (\ref{myform}) suggests the way of inverting the
enhanced operator $\Lambda^*-(\VV^*\cdot\nabla)$:
\begin{itemize}
\item {\em Step 1}: Discard the enhanced free flight operator. The resulting local in space linear integral equation, $\Lambda^*[\varphi]=D$, is similar to the Chapman-Enskog equation (\ref{CE1Boltzmann}), and has unique solution by the Fredholm alternative:
\begin{equation}
\varphi_{\rm loc}(\xx)=\left(\Lambda_{\textbf \xx}^*\right)^{-1}[D(\xx)].
\end{equation}
Here we have explicitly indicated the space variables in order to stress the fact of locality. (For a given $\xx$, both $D(\xx)$ and $\varphi_{\rm loc}(\xx)$ are functions of $\xx$ and $\vv$ and $\Lambda_{\textbf \xx}$ is an integral in $\vv$ operator.)
\item {\em Step 2}: Fourier-transform the local solution:
\begin{equation}
\hat{\varphi}_{\rm loc}(\kk)=\int e^{-i{\textbf \kk\cdot \textbf \xx}}\varphi_{\rm loc}(\xx) \D \xx.
\end{equation}
\item {\em Step 3}: Replace the Fourier-transformed enhanced free flight operator with its main symbol and solve the linear integral equation:
\begin{equation}
[\Lambda_{\textbf  \xx}^*+i(\VV_{\textbf  \xx}^*\cdot\kk)][\hat{\varphi}(\xx,\kk)]=\hat{D}(\xx,\kk),
\end{equation}
where
\begin{equation}
\hat{D}(\xx,\kk)=\Lambda_{\textbf  \xx}^*[\hat{\varphi}_{\rm loc}(\kk)].
\end{equation}
\item {\em Step 4}: Back-transform the result:
\begin{equation}
\varphi=(2\pi)^{-3}\int e^{i{\textbf \kk\cdot \textbf \xx}}\hat{\varphi}(\xx,\kk) \D \kk;
\end{equation}
the resulting $\varphi$ is a function of $\xx$ and $\vv$.
\end{itemize}
Several comments are in order here. The above approach to solving the invariance equation
is the realization of the Fourier integral operator and parametrix expansion techniques
\cite{Shubin2001,Treves1980}. The equation appearing in Step 3 is in fact the first term of
the parametrix expansion. At each step of the algorithm, one needs to solve linear
integral equations of the type familiar from the standard literature on the Boltzmann
equation. Solutions at each step are unique by the Fredholm alternative. In practice, a
good approximation for such linear integral equations is achieved by a projection on a
finite-dimensional basis.  Even with these approximations, evaluation of the correction to
the local Maxwellian remains rather involved. Nevertheless, several results in limiting
cases were obtained, and are reviewed below.

\begin{figure}
\centering{
\includegraphics[width=0.5 \textwidth]{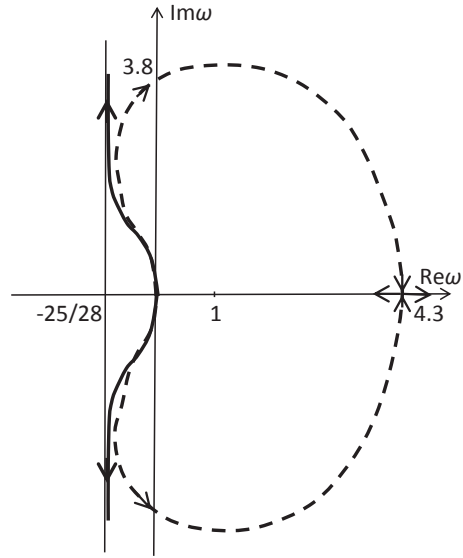}
\caption {\label{FigDisp} Acoustic dispersion curves for the frequency-response nonlocal
approximation (\ref{4.70}) (solid line) and for the Burnett
approximation of the  Chapman-Enskog expansion  \cite{Bob} (dashed
line). Arrows indicate the direction of increase of $k^{2}$.}}
\end{figure}

For the unidirectional flow near the global equilibrium ($n=n_0, \uu=0, T=T_0$) for Maxwell's molecules  the
iteration gives the following expressions for the $xx$ component of the stress tensor
$\sigma$ and the $x$ component of the heat flux $q$ for 1D solutions (in the
corresponding dimensionless variables (\ref{AnnPhysvar})):
\begin{equation}\boxed{
\begin{aligned}
\sigma & =  -\frac{2}{3}n_{0}T_{0}\left( 1 - \frac{2}{5}{\partial_x ^{2}} \right)^{-1}
 \left(2{\partial_x u}-
3 {\partial_x ^{2}T} \right);  \\
q& = -\frac{5}{4}n_{0} T^{3/2}_{0}\left( 1 - \frac{2}{5}{\partial_x ^{2}} \right)^{-1}
\left( 3  {\partial_x T}-\frac{8}{5}{\partial_x ^{2}u} \right) . \label{4.70}
\end{aligned}}
\end{equation}
The corresponding dispersion curves are presented in Fig.~\ref{FigDisp}, where the
saturation effect is obvious.

Already at the first iteration the nonlinear terms are strongly coupled with the nonlocality
in expressions for $\sigma$ in $q$ (see \cite{GKTTSP94,GorKarLNP2005}).
Viscosity tends to positive infinity for high speed of compression (large negative ${\rm div}
\uu$). In other words, the  flow  becomes  ``infinitely viscous" when ${\partial_x u}$
approaches the critical negative value $-u_x^*$. This  infinite viscosity
threshold prevents  a transfer of  the   flow   into nonphysical region of negative
viscosity if ${\partial_x u} < -u_x^*$  because
of the  ``infinitely  strong  damping"  at $-u_x^*$.

The large {\it positive} values of $\partial_x u$ means that the gas diverges rapidly,
and the flow becomes nonviscid because the particles  retard  to exchange their momentum.
On the contrary, its {\it negative} values (near $-u^*_x$ ) describe an extremely  strong
compression of the flow, which results in a `solid jet' limit with an infinite
viscosity \cite{GKsolidliq1995}.

As an example, we present the result of the above micro-local analysis for the part of the stress tensor
$\sigma$ which does not vanish when $T$ and $n$ are fixed:
\begin{equation}\boxed{
\begin{aligned}
\sigma(x)&=-\frac{1}{ 6 \pi } n(x) \int^{+\infty}_{-\infty } \D y\int^{+\infty }_{-\infty}  \D k \exp
(ik(x-y))\frac{2}{3} {\partial_y u(y)} \\
&\times  \left[ \left( n(x)\lambda _{3}+\frac{11}{9} {\partial_x u(x)} \right) \left(n(x)\lambda _{4}+\frac{27}{ 4} {\partial_x u(x)} \right) +
\frac{k^2 v^2_{T}(x)}{9} \right]^{-1} \\
&\times  \left[ \left( n(x)\lambda _{3}+\frac{11}{ 9}
{\partial_x u(x)} \right) \left(n(x)\lambda _{4}+\frac{27}{ 4} {\partial_x u(x)} \right)
 \right.
\\ & \left. + \frac{4}{ 9} \left(n(y)\lambda _{4}+\frac{27}{ 4} {\partial_y u(y)} \right)
v^{-2}_{T}(x)(u(x)-u(y))^{2}{\partial_x u(x)}\right. \\
&\left. -\frac{2}{ 3}ik(u(x)-u(y)){\partial_x u(x)}  \right]  \left(n(y)\lambda _{3}+\frac{11}{9} {\partial_y u(y)} \right)^{-1}  \\
& + O \left( {\partial_x \ln T(x)}, {\partial_x \ln  n(x)} \right).
\label{4.74b}
\end{aligned}}
\end{equation}
The answer in this form does not depend on the detailed collision model. Only the general
properties like conservation laws, $H$-theorem and Fredholm's alternative for the
linearized collision integral are used. All the specific information about the collision
model is collected in the positive numbers $\lambda_{3,4}$. They are  represented by
quadratures in \cite{GKTTSP94,GorKarLNP2005}. The `residual' terms describe the part of
the stress tensor governed by the temperature and density gradients.

The simplest local approximation to this singularity in $\sigma$ has the form
\begin{equation}\label{LocSingSigma}\boxed{
\sigma =-\mu_0(T) n \left(1+\frac{\partial_x u}{u_x^*} \right)^{-1} {\partial_x u } .}
\end{equation}

For the viscosity factor $R$ (\ref{R}) this approximation gives (compare to
(\ref{viscosity}) and Fig.~\ref{AnnPhysFig10}).
\begin{equation}
R=\frac{const}{1+{\partial_x u}/{u_x^*} }
\end{equation}

The approximations with singularities similar to (\ref{LocSingSigma}) with $u_x^*=3/7$
have been also obtained by the partial summation of the Chapman--Enskog series
\cite{GKJETP91,GKTTSP92}.

As we can see,  the invariance correction results in a strong  coupling between  non-locality and non-linearity, and is
far from the conventional Navier--Stokes and Euler equation or other truncations of the
Chapman--Enskog series. Results of the micro-local correction to the local Maxwellian are quite similar to the summation of the
selected main terms of the Chapman--Enskog expansion.
In general, the question about the hydrodynamic invariant manifolds for the
Boltzmann equation remains less studied so far because the coupling between the non-linearity and the
non-locality brings about new challenges in calculations and proofs.
There is hardly a reason to expect that
the invariant manifolds for the genuine Boltzmann equation will have a nice analytic form
similar to the exactly solvable reduction problem for the linearized Grad equations.
Nevertheless, some effects persist: the saturation of dissipation for high frequencies
and the nonlocal character of the hydrodynamic equations.

\section{The projection problem and the entropy equation\label{Sec:Projector}}

The exact invariant manifolds inherit many properties of the original systems:
conservation laws, dissipation inequalities (entropy growth) and hyperbolicity of the
exactly reduced system follow from these properties of the original system. The
reason for this inheritance is simple: the vector field of the original system is tangent
to the invariant manifold and if $M(t)$ is a solution to the exact hydrodynamic equations
then, after the lifting operation, $f_{M(t)}$ is a solution to the original kinetic
equation.

In real-world applications, we very rarely meet the exact reduction from kinetics to
hydrodynamics and should work with the {\em approximate} invariant manifolds. If $\ff_M$
is not an exact invariant manifold then a special {\em projection problem} arises
\cite{GK1,UNIMOLD,Radkevich2006}: how should we define the projection of the vector field
on the manifold $\ff_M$ in order to preserve the most important properties, the
conservation laws (first law of thermodynamics) and the positivity of entropy production
(second law of thermodynamics). For hydrodynamics, the existence of the `natural' moment
projection $m$ (\ref{BoltzMomentProj}) masks the problem.

The problem of dissipativity preservation
attracts much attention in the theory of shock waves.
For strong shocks it is necessary to use the kinetic representation, for rarefied gases
the Boltzmann kinetic equation gives the framework for studying the structure of strong
shocks \cite{Cercignani}.  One of the common heuristic ways to use the Boltzmann equation
far from local equilibrium consists of three steps:
\begin{enumerate}
\item{Construction of a specific ansatz for the distribution function for a given
    physical problem;}
\item{Projection of the Boltzmann equation on the ansatz;}
\item{Estimation and correction of the ansatz (optional).}
\end{enumerate}
The first and, at the same time, the  most successful ansatz for the distribution
function in the shock layer was invented in the middle of the twentieth century. It is
the bimodal Tamm--Mott-Smith approximation (see, for example, the book
\cite{Cercignani}):
\begin{equation}\label{TMS}
f(\vv, \xx)=f_{\rm TMS}(\vv, z) = a_-(z)f_-(\vv)+a_+(z)f_+(\vv),
\end{equation}
where $z$ is the space coordinate in the direction of the shock wave motion, $f_\pm(\vv)$
are the downstream and the upstream Maxwellian distributions, respectively. The
macroscopic variables for the Tamm--Mott-Smith approximation are the coefficients
$a_{\pm}(z)$, the lifting operation is given by (\ref{TMS}) but is remains unclear how to
project the Boltzmann equation onto the linear manifold (\ref{TMS}) and create the
macroscopic equation.

To respect second law of thermodynamics and provide positivity of entropy
production, Lampis \cite{Lampis} used the entropy density $s$ as a new variable.
The entropy density is defined as a functional of $f(\vv)$, $s(\xx)=-\int f(\xx,\vv) \ln f(\xx,\vv) \, \D^3 \vv$.
For each distribution $f$ the time derivative of $s$ is defined by the Boltzmann equation and the chain rule:
\begin{equation}\label{entropyEq}
\partial_t s= -\int  \ln f \partial_t f \, \D^3 \vv=\mbox{ entropy flux }+\mbox{ entropy production }.
\end{equation}
The distribution $f$ in (\ref{entropyEq}) is defined by the Tamm--Mott-Smith approximation:
\begin{enumerate}
\item{Calculate the density $n$ and entropy density $s$ on the Tamm--Mott-Smith approximation (\ref{TMS}) as
functions of $a_\pm$, $n=n(a_+,a_-)$, $s=s(a_+,a_-)$; }
\item{Find the inverse transformation $a_{\pm}(n,s)$;}
\item{The lifting operation in the variables $n$ and $s$ is
$$f_{(n,s)}(\vv)=a_-(n,s)f_-(\vv)+ a_{+}(n,s)f_+(\vv).$$}
\end{enumerate}
This combinational of the natural projection (\ref{entropyEq}) and the Tamm--Mott-Smith
lifting operation provides the approximate equations on the Tamm--Mott-Smith manifold
with positive entropy production. Several other projections have been tested
computationally \cite{Naka}. All of them violate second law of thermodynamics because
for some initial conditions the entropy production for them becomes negative at some
points. Indeed, introduction of the entropy density as an independent variable with the
natural projection of the kinetic equation on this variable seems to be an attractive and
universal way to satisfy the second law of thermodynamics on smooth solutions but
near the equilibria this change of variables becomes singular.

Another universal solution works near equilibria (and local equilibria). The advection operator does not change
entropy. Let us consider a linear approximation to a space--uniform kinetic equation near
equilibrium $f^*(\vv)$: $\partial_t \delta f =K f$. The second differential of entropy
generates a positive quadratic form
\begin{equation}\label{EntLinProd}
\langle \varphi,\psi\ \rangle_{f^*}=-(D^2 S) (\varphi,\psi)=\int \frac{\varphi \psi}{f^*}\, \D^3 \vv .
\end{equation}
The quadratic approximation to the entropy production is non-negative:
\begin{equation}\label{posentprolin}
-\langle \varphi,K \varphi\ \rangle_{f^*}\geq 0.
\end{equation}
Let $T$ be a closed linear subspace in the space of distributions. There is a unique
projector $P_T$ onto this subspace which does not violate the positivity of entropy
production for any bounded operator $K$ with property (\ref{posentprolin}): If $-\langle
P_T \varphi, P_T K P_T \varphi \rangle_{f^*}\geq 0$ for all $\varphi$, $\psi$ and all
bounded $K$ with property (\ref{posentprolin}), then $P_T$ is an orthogonal projector
with respect to the entropic inner product  (\ref{EntLinProd}) \cite{InChLANL,UNIMOLD}.
This projector acts on functions of $\vv$. For a local equilibrium $f^*(\xx,\vv)$ the projector
is constructed for each $\xx$ and acts on functions $\varphi(\xx,\vv)$ point-wise at each point $\xx$.
Liu and Yu \cite{LiuYu2004CMP} also used this projector in a vicinity of local equilibria
for the micro--macro decomposition in the analysis of the shock profiles and for the
study nonlinear stability of the global Maxwellian states \cite{LiuYu2004PhD}. Robertson
studied the projection onto manifolds constructed by the conditional maximization of the entropy
and the micro--macro decomposition in the vicinity of such manifolds \cite{Robertson}. He obtained the
orthogonal projectors with respect to the entropic inner product and called this result ``the
equation of motion for the generalized canonical density operator''.

The general case can be considered as a `coupling' of the above two examples: the introduction of
the entropy density as a new variable, and the orthogonal projector with respect to entropic inner
product. Let us consider all smooth vector fields with non-negative entropy production.
The projector which preserves the nonnegativity of the entropy production for {\em
all} such fields turns out to be unique. This is the so-called {\em thermodynamic projector}
\cite{GK1,InChLANL,UNIMOLD,GorKarLNP2005}. Let us describe this projector $P$ for a given state $f$,
closed subspace $T_f = {\rm im P_T}$, and the differential $(DS)_f$ of the entropy $S$ at
$f$. For each state $f$ we use the entropic inner product (\ref{EntLinProd}) at
 $f^*=f$. There exists a unique vector $g(f)$ such that $\langle g, \varphi
\rangle_f =(DS)_f(\varphi )$ for all $\varphi $. This is nothing but the Riesz representation of the
linear functional $D_xS$ with respect to entropic scalar product. If $g \neq 0$ then the
thermodynamic projector of the vector field $J$ is
\begin{equation}\label{TermProj}\boxed{
P_T(J)=P^{\bot} (J) + \frac{ g^{\|}}{\langle g^{\|}|g^{\|}\rangle_{f}} \langle
g^{\bot}|J \rangle_{f},}
\end{equation}
where $P_T^{\bot}$ is the orthogonal projector onto $T_f$ with respect to the entropic
scalar product, and the vector $g$ is split onto tangent and orthogonal components:
$$g = g^{\|} + g^{\bot}; \; g^{\|}= P^{\bot} g;  \: g^{\bot} =
(1-P^{\bot}) g .$$ This projector is defined if  $g^{\|}\neq 0$. If $g^{\|} = 0$ (the
equilibrium point) then $J=0$ and  $P(J)= P^{\bot} (J)=0$.

\begin{figure}[t]
\begin{centering}\boxed{
\includegraphics[width=0.8 \textwidth]{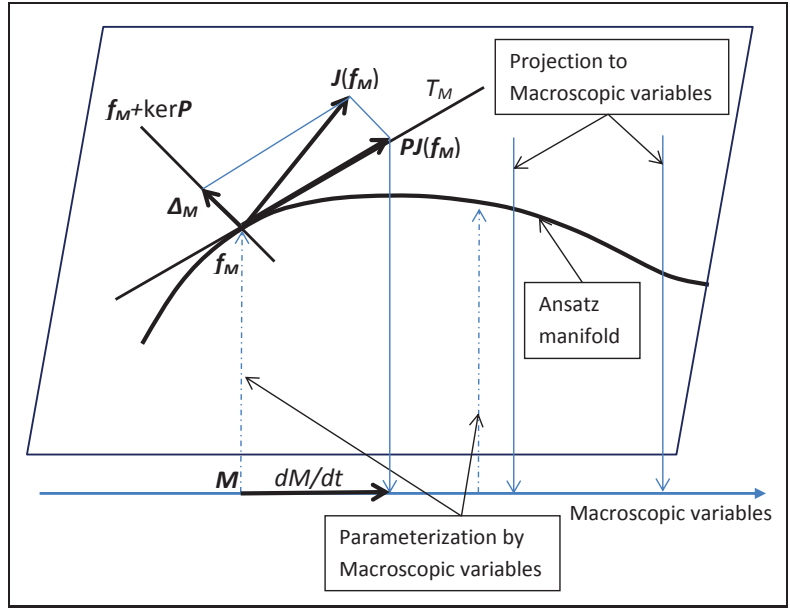}}
\caption {\label{InManProj} The main geometrical structures of model
reduction with an approximate invariant manifold (the {\em ansatz} manifold):  $J(f)$ is the vector field of
the system under consideration, $\partial_t f = J(f)$,  the lifting map $M \mapsto \ff_M$ maps
a macroscopic field $M$ into the corresponding point $\ff_M$ on the
ansatz manifold, $T_M$ is the tangent space to the ansatz manifold
at point $\ff_M$, $P$ is the thermodynamic projector onto $T_M$ at point $\ff_M$,
$PJ(\ff_M)$ is the projection of the vector
$J(\ff_M)$ onto tangent space $T_M$, the vector field $\D M / \D t$
describes the induced dynamics on the space of macroscopic variables, $\Delta_M
= (1-P)J(\ff_M)$ is the defect of invariance, the affine subspace $\ff_M+
\ker P$ is the plane of fast motions,  and $\Delta_M \in \ker P$.
The invariance equation is $\Delta_M=0$.}
\end{centering}
\end{figure}

The selection of the projector in the form (\ref{TermProj}) guaranties preservation of
entropy production.  The thermodynamic projector can be applied for the projection of the
kinetic equation onto the tangent space to the approximate invariant manifold if the
differential of the entropy does not annihilate the tangent space to this manifold.
(Compare to the relative entropy approach in \cite{Saint-RaymondCompanion}.)

Modification of the projector changes also the a simplistic picture of the separation of motions
(Fig.~\ref{fig1SLw}). The modified version is presented in Fig.~\ref{InManProj}. The main
differences are:
\begin{itemize}
\item{The projection of the vector field $J$ on the macroscopic variables $M$  goes
    in two steps, $$J(\ff_M) \mapsto PJ(\ff_M)\mapsto m(PJ(\ff_M)),$$ the first
    operation $J(\ff_M) \mapsto PJ(\ff_M)$ projects $J$ onto the tangent plane $T_M$ to
    the ansatz manifold  at point $\ff_M$ and the second
    is the standard projection onto macroscopic variables
    $m$. Therefore, the macroscopic equations are $\partial_t M=m(PJ(\ff_M))$ instead
    of (\ref{projectedMIM}).}
\item{The plane of {\em fast motion} is now $\ff_M + \ker P$ instead of  $\ff_M +
    \ker m$ from Fig.~\ref{fig1SLw}. }
\item{The entropy maximizer on $\ff_M + \ker P$  is $\ff_M$, exactly as the local
    Maxwel\-lians $\ff_M^{\rm LM}$ are the entropy maximizers on $\ff_M^{\rm LM}+ \ker
    m$. Thus, the entropic projector allows us to represent an ansatz manifold as a
    collection of the conditional entropy maximizers.}
\end{itemize}

For  details of the thermodynamic projector construction we refer to \cite{UNIMOLD,GorKarLNP2005}. Some examples with
construction of the thermodynamic projector with preservation of linear conservation laws
are presented in \cite{GorGorKar2004}.

Another possible modification is a modification of the entropy functional. Recently, Grmela
\cite{GrmelaAdv2010,GrmelaCAMWA2013} proposed to modify the entropy functional after each
step of the Chapman--Enskog expansion in order to transform the approximate invariant
manifold into the manifold of the conditional entropy maximizers. This idea is very similar to
the thermodynamic projector in the following sense: any point $\varphi$ on the
approximate invariant manifold is the conditional entropy maximum on the linear manifold
$\varphi+\ker P_T$, where $T=T_{\varphi}$ is the tangent subspace to the manifold at
point $\varphi$. Both modifications represent the approximate invariant manifold as a set
of conditional maximizers of the entropy.

\section{Conclusion}

{\bf It is useful to solve the invariance equation.}  This is a particular case of the
Newton's famous  sentence: ``It is useful to solve differential equations'' (``Data
{\ae}quatione quotcunque fluentes quantit{\ae} involvente fluxiones invenire et vice
versa,'' translation published by V.I. Arnold \cite{ArnoldGeomMet1983}). The importance of
the invariance equation has been recognized in mechanics by Lyapunov in  his thesis
(1892) \cite{Lya1992}.  The problem of persistence and bifurcations of invariant
manifolds under perturbations is one of the most seminal problems in dynamics
\cite{KAM,KAM1,KAMAr2,KAMAr3,HiPuSh,IneManTe88}.

Several approaches to the computation of invariant manifolds have been developed: Lyapunov
series \cite{Lya1992}, methods of geometric singular perturbation theory
\cite{Fenichel0,Fenichel9,JonesGSP1995} and various power series expansions
\cite{Ens,Chapman,Beyn1}. The graph transformation approach was invented by Hadamard in
1901 \cite{Hadamard1901} and developed further by many authors
\cite{HiPuSh,Fontish1990,GuVlad2004,Kornev2007}. The Newton-type direct iteration methods
in various forms \cite{RouFra,GK1,GorKarTTSPreg1993,GKTTSP94,Lam} proved their efficiency
for model reduction and calculation of slow manifolds in kinetics. There is also a series
of numerical methods based on the analysis of motion of an embedded manifold along the
trajectories with subtraction of the motion of the manifold `parallel to themselves'
\cite{Kev,CMIM,NafeMaas2002,GorKarLNP2005}.

The Chapman--Enskog method \cite{Ens,Chapman} was proposed in 1916. This method aims
to construct the invariant manifold for the Boltzmann equation in the form of a series in
powers of a  small parameter, the Knudsen number $K\!n$. This invariant manifold is
parameterized by the hydrodynamic fields (density, velocity, temperature). The
zeroth-order term of this series is the corresponding local equilibrium. This form of the
solution (the power series and the local equilibrium zeroth term) is, at the same time, a
selection rule that is necessary to choose the hydrodynamic (or Chapman--Enskog) solution
of the invariance equation.

If we truncate the Chapman--Enskog series at the zeroth term then we get the Euler
hydrodynamic equations, the first term gives the Navier-Stokes hydrodynamics but already
the next term (Burnett) is singular and gives negative viscosity for large divergence of the
flow and instability of short waves. Nevertheless, if we apply, for example, the
Newton--Kantorovich method \cite{GKTTSP94,CMIM,GorKarLNP2005} then all these
singularities vanish (Sec.~\ref{IMBoltzmann}).

The Chapman--Enskog expansion appears as the Taylor series for the solution of the
invariance equation. Truncation of this series may approximate the hydrodynamic
invariant manifold in some limit cases such as the long wave limit or a vicinity of the
global equilibrium. Of course, the results of the invariant manifold approach should
coincide with the proven hydrodynamic limits of the Boltzmann kinetics
\cite{BardosGolLever1991,LionsMasmud1999,GolseS-R2004,Saint-RaymondBook,Saint-RaymondCompanion}
`at the end of relaxation'.

In general, there is no reason to hope that a few first terms of the Taylor series
give an appropriate global approximation of solutions of the invariance equation
(\ref{generalIE}). This is clearly demonstrated by the exact solutions
(Sec.~\ref{Sec:Exact}, \ref{Sec:ExactLinKin}).

The invariant manifold idea was present implicitly in the original Enskog and Chapman works and in
most subsequent publications and textbooks. An explicit formulation of the invariant
manifold programme for the derivation of fluid mechanics and hydrodynamic limits from the
Boltzmann equation was published by McKean \cite{McKean1965} (see Fig.~\ref{McKeanDiag}
in Sec~\ref{Sec:IdeaIMkinetics}). At the same time, McKean noticed that the problem
of the invariant manifold for kinetic equations does not include the small parameter because
by the rescaling of the space dependence of the initial conditions we can remove the
coefficient in front of the collision integral: there is no difference between the
Boltzmann equations with different $K\!n$. Now we know that the formal `small' parameter
is necessary for the selection of the hydrodynamic branch of the solutions of the
invariance equation because this equation can have many more solutions. (For example,
Lyapunov used for this purpose analyticity of the invariant manifold and selected the
zeroth approximation in the form of the invariant subspace of the linear approximation.)

The simplest example of  invariant manifold is a trajectory (invariant curve). Therefore, the method of invariant manifold
may be used for the construction and analysis of the trajectories. This simple idea is useful and the method of invariant manifold was applied
for solution of the following problems:
\begin{itemize}
\item{For analysis and correction of the Tamm--Mott-Smith approximation of  strong
    shock waves far from local equilibrium \cite{GK1}, with the Newton iterations for
    corrections;}
\item{For analysis of reaction kinetics \cite{ChiGorKar2007} and reaction--diffusion
    equations \cite{MenPow2013};}
\item{For lifting of shock waves from the piece-wise solutions of the Euler equation to the solutions
of the Boltzmann equation hear local equilibrium for small $K\!n$ \cite{LiuYu2004CMP};}
\item{For analytical approximation of the relaxation trajectories \cite{GorKarNonZmi1996}. (The me\-thod is tested
for the space-independent Boltzmann equation with various collisional mechanisms.)}
\end{itemize}

{\bf The invariant manifold approach to the kinetic part of the 6th Hilbert's Problem}
concerning the limit transition from the Boltzmann kinetics to mechanics of continua was
invented by Enskog almost a century ago, in 1916 \cite{Ens}. From a physical perspective,
it remains the main method for the construction of macroscopic dynamics from dissipative
kinetic equations. Mathematicians, in general, pay less attention to this approach
because usually in its formulation the solution procedure (the algorithm for the
construction of the bulky and singular Chapman--Enskog series) is not separated from the
problem statement (the hydrodynamic invariant manifold). Nevertheless, since the 1960s
the invariant manifold statement of the problem has been clear for some researchers
\cite{McKean1965,GKTTSP94,GorKarLNP2005}.

Analysis of the simple kinetic models with algebraic hydrodynamic invariant manifolds
(Sec.~\ref{Sec:Exact}) shows that the {\em hydrodynamic invariant manifolds may exist
globally} and the divergence of the Chapman--Enskog series does not mean the
non-existence or non-analyticity of this manifold.

The invariance equation for the more complex Grad kinetic equations (linearized) is also
obtained in an algebraic form (see Sec.~\ref{Sec:GradDestroy} and
\cite{KGAnPh2002,GorKarLNP2005} for 1D and Sec.~\ref{Sec:Grad3DDestroy} and
\cite{ColKarKroe2007_2} for 3D space). An analysis of these polynomial equations shows that
the real-valued solution of the invariance equation in the $k$-space may break down for very
short waves. This effect is caused by the so-called entanglement of hydrodynamic and
non-hydrodynamic modes.

The linearized equation with the BGK collision model \cite{BGK} includes the genuine free
flight advection operator and is closer to the Boltzmann equation in the hierarchy of
simplifications. For this equation, there are numerical indications that the hydrodynamic modes are separated from the
non-hydrodynamic ones and the calculations show that the hydrodynamic invariant manifold may
exist globally (for all values of the wave vector $k$) \cite{KarColKroe2008PRL}.

It seems more difficult to find a nonlinear Boltzmann equation with exactly solvable
invariance equation and summarize the Chapman--Enskog series for a nonlinear kinetic
equation exactly. Instead of this, we select in each term of the series the terms of the
main order in the power of the Mach number $M\!a$ and exactly summarize the resulting
series for the simple nonlinear 1D Grad system (Sec.~\ref{Sec:NlinVis},
\cite{KGAnPh2002,GorKarLNP2005}). This expansion gives the dependence of the viscosity on the
velocity gradient (\ref{R}), (\ref{AnnPhysinvariance}), (\ref{viscosity}).

The exact hydrodynamics projected from the invariant manifolds inherits many useful
properties of the initial kinetics: conservation laws, dissipation inequalities, and (for
the bounded lifting operators) hyperbolicity (Sec.~\ref{Sec:Hyperbol}). Also, the
existence and uniqueness theorem  may be valid in the projections if it is valid for the
original kinetics. In applications, for the approximate hydrodynamic invariant manifolds,
the projected equation may violate many important properties. In this case, the change of
the projector  operator solves some of these problems  (Sec.~\ref{Sec:Projector}). The
construction of the thermodynamic projector guarantees the positivity of entropy production
even in very rough approximations \cite{GK1,UNIMOLD}.

{\bf At the present time, Hilbert's 6th Problem is not completely solved in its kinetic
part.} More precisely, there are several hypotheses we can prove or refute. {\em The
Hilbert hypothesis has not been unambiguously formulated} but following his own works in
the Boltzmann kinetics we can guess that he expected to receive the Euler and
Navier--Stokes equations as an ultimate hydrodynamic limit of the Boltzmann equation.

Now, the Euler limit is proven for the  limit  $K\!n, M\!a \to 0$, $M\!a \ll K\!n$, and
the Navier--Stokes limit is proven for $K\!n, M\!a \to 0$, $M\!a \sim K\!n$. In these
limits, the flux is extremely slow and the gradients are extremely small (the velocity,
density and temperature do not change significantly over a long distance). The system is
close to the global equilibrium. Of course, after rescaling these solutions restore some
dynamics but this rescaling erases some physically important effects. For example, it is
a simple exercise to transform an attenuation curve with saturation from
Fig.~\ref{Attenuation} into a parabola (Navier--Stokes) or even into a horizontal straight
line (no attenuation, the Euler limit) with arbitrary accuracy by the rescaling of space
and time.

We can state at present that beyond this limit, the Euler and Navier--Stokes
hydrodynamics do not provide the proper hydrodynamic limit of the Boltzmann equation. A
solution of the Boltzmann equation relaxes to the equilibrium \cite{DescvilVill2005} and,
on its way to equilibrium, the classical hydrodynamic limit will be achieved as an
intermediate asymptotic (after the proper rescaling). This recently proven result fills
an important gap in our knowledge about the Boltzmann equation but from the physics
perspective this is still the limit $K\!n, M\!a \to 0$ (with the proof that this limit
will be achieved on the path to equilibrium).

{\em The invariant manifold hypothesis} was formulated clearly by McKean
\cite{McKean1965} (see Sec.~\ref{Sec:IdeaIMkinetics}, Fig.~\ref{McKeanDiag} and
Sec.~\ref{Sec:e=1Cat}): the kinetic equation admits  an invariant manifold parameterized
by the hydrodynamic fields, and the Chapman--Enskog series are the Taylor series for this
manifold. {\em Nothing is expected to be small and no rescaling is needed}. After the
publication of the McKean work (1965), this hypothesis was supported by exactly solved
reduction problems, explicitly calculated algebraic forms of the invariance equation and
direct numerical solutions of these equations for some cases like the linearized BGK
equation.

In addition to the existence of the hydrodynamic invariant manifold some stability
conditions of this manifold are needed in practice. Roughly speaking, the relaxation to
this manifold should be faster than the motion along it. An example of such a condition
gives the separation of the hydrodynamic and non-hydrodynamic modes for linear kinetic
equations (see the examples in Sec.~\ref{Sec:ExactLinKin} and \ref{Sec:Exact}). It should
be stressed that the strong separation of the relaxation times (Fig.~\ref{fig1SLw}) is
impossible without a small parameter. For the ``$\varepsilon=1$ approach'' we can expect only
some dominance of the relaxation towards the hydrodynamic manifold over the relaxation
along it.

{\em The capillarity hypothesis} was proposed very recently by Slemrod
\cite{SlemQuaterly2012,SlemCAMWA2013}. He advocated the $\varepsilon=1$ approach and
studied the exact sum of the Chapman--Enskog series obtained in
\cite{GKPRL96,KGAnPh2002}. Slemrod  demonstrated that in the balance of
the kinetic energy (\ref{EnergyXfin}) a viscosity term appears (\ref{CapillIdea}) and the
saturation of dissipation can be represented as the interplay between viscosity and
capillarity (Sec.~\ref{Sec:caplillarity}).

On the basis of this idea and some heuristics about the relation between the moment (Grad)
equations and the genuine Boltzmann equation, Slemrod suggested that the proper exact
hydrodynamic equation should have the form of the Korteweg hydrodynamics
\cite{Korteweg1901,Dunn1985,Slem1} rather than of Euler or Navier--Stokes ones.

The capillarity--like terms appear, indeed, in the energy balance for all hydrodynamic
equations found as a projection of the kinetic equations onto the exact or approximate
invariant hydrodynamic manifolds. In that (`wide') sense, the capillarity hypothesis is
plausible. In the more narrow sense, as does the validity of the Korteweg hydrodynamics, the
capillarity hypothesis requires some efforts for reformulation. The interplay between
nonlinearity and nonlocality on the hydrodynamic manifolds seems to be much more complex
than in the Korteweg equations (see, for example, Sec.~\ref{IMBoltzmann}, eq.
(\ref{4.74b}), or \cite{GKTTSP94,GorKarLNP2005}). For a serious consideration of this
hypothesis we have to find out for which asymptotic assumption we expect it to be valid (if
$\varepsilon=1$ then this question is non-trivial).

{\bf In the context of the exact solution of the invariance equations, three problems
become  visible:}
\begin{enumerate}
\item{To prove the existence of the hydrodynamic invariant manifold for the
    linearized Boltzmann equation.}
\item{To prove the existence of the analytic hydrodynamic invariant manifold for the
    Boltzmann equation.}
\item{To match the low-frequency, small gradient asymptotics of the invariant
    manifold with the high-frequency, large gradient asymptotics and prove the
    universality of the matched asymptotics in some limits.}
\end{enumerate}

The first problem seems to be not extremely difficult. For its positive solution, the
linearized collision operator should be bounded and satisfy the spectral gap condition.

For the nonlinear Boltzmann equation, the existence of the analytic invariant manifold seems
to be plausible but the singularities in the first Newton--Kantorovich approximation
(Sec.~\ref{IMBoltzmann}) may give a hint about the possible difficulties in the highly
non-linear regions. In this first approximation,  flows with very high negative
divergence cannot appear in the evolution of flows with lower divergence because the
viscosity tends to infinity. This  `solid jet' \cite{GKsolidliq1995} effect can be
considered as a sort of phase transition.

The idea of an exact hydrodynamic invariant manifold is  attractive and the approximate
solutions of the invariance equation can be useful but the possibility of elegant
asymptotic solution is very attractive too. Now we know that we do not know how to state
the proper problem. Can the observable hydrodynamic regimes be considered as solutions of
a simplified hydrodynamic equation? Here a new, yet non-mathematical notion appears, ``the
observable hydrodynamic regimes''. We can speculate now, that when the analytic invariant
manifold exists, then together with the low-frequency, low-gradient Chapman--Enskog
asymptotics the high--frequency and high--gradient asymptotics of the hydrodynamic
equations are also achievable in a constructive simple form (see examples in
Sec.~\ref{HighFreqAs} and Sec.~\ref{Sec:NlinVis}). The bold hypothesis \#3 means that in
some asymptotic sense only the extreme cases are important and the behavior of the
invariant manifold between them may be substituted by matching asymptotics. We still do
not know an exact formulation of this hypothesis and can only guess how the behaviour of
the hydrodynamic solutions becomes dependent only on the extreme cases. Some hints may be
found in recent works about the universal asymptotics of solutions of PDEs with small
dissipation \cite{Dubrovin2012} (which develop the ideas of Il'in proposed in the analysis of
boundary layers \cite{Ilin1992}).

We hope that problem \#1 about the existence  of hydrodynamic invariant manifolds for the
linearized Boltzmann equation will be solved soon, problem \#2 about the full
nonlinear Boltzmann equation may be approached and solved after the first one. We expect
that the answer will be positive: hydrodynamic invariant manifolds do exist under the
spectral gap condition.

Once the first two problems will be solved then the entire object, the hydrodynamic
invariant manifold will be outlined. For this manifold,  the various asymptotic
expansions could be produced, for low frequencies and gradients, for high frequencies,
and for large gradients. Matching of these expansions and analysis of the resulting
equations may give material for the exploration of hypothesis \#3. Some guesses about the
resulting equations may be formulated now, on the basis of the known results. For example
we can expect that non-locality may be reduced to the substitution of the time derivative
$\partial_t$ in the system of fluid dynamic equations by $(1-W\Delta)\partial_t$, where
$\Delta$ is the Laplace operator and $W$ is a positive definite matrix (compare to
Sec.~\ref{HighFreqAs}). It seems interesting and attractive that the resulting equations
may be new and, at the same time, simple and beautiful hydrodynamic equations.

From the mathematical perspective, the approach based on the invariance equation now creates
more questions than answers. It changes the problem statement and the exact solutions
give us some hints about the possible answers.

\section*{Acknowledgements}
We are grateful to M. Gromov for stimulating discussion, to M. Slemrod for inspiring
comments and ideas and to L. Saint-Raymond for useful comments.
IK gratefully acknowledges support by the European Research Council
(ERC) Advanced Grant 291094-ELBM.

\section*{About the authors}
Professor Alexander N. Gorban holds a personal chair in Applied
Mathematics at the University of Leicester since 2004. He  worked
for Russian Academy of Sciences, Siberian Branch (Krasnoyarsk, Russia)
and ETH Z\"urich (Switzerland), was a visiting professor and
research scholar at Clay Mathematics Institute (Cambridge, US), IHES
(Bures--sur-Yvette, \^Ile de France), Courant Institute of
Mathematical Sciences (NY, US) and Isaac Newton Institute for
Mathematical Sciences (Cambridge, UK). Main research interests:
Dynamics of systems of physical, Chemical and biological kinetics;
Biomathematics; Data mining and model reduction problems.

Professor Ilya Karlin is Faculty Member at the Department of Mechanical and
Process Engineering, ETH Zurich, Switzerland. He was Alexander von Humboldt
Fellow at the University of Ulm (Germany), CNR Fellow at the Institute of
Applied Mathematics CNR ``M. Picone'' (Rome, Italy), and Senior Lecturer
in Multiscale Modeling at the University of Southampton (England).
Main research interests: Exact and non-perturbative results in kinetic theory;
Fluid dynamics; Entropic lattice Boltzmann method;
Model reduction for combustion systems.


\end{document}